\documentclass[a4paper,UKenglish,cleveref, autoref, thm-restate]{article}

\usepackage[a4paper, left=3.5cm, right=3.5cm]{geometry}
\usepackage{authblk}
\usepackage[colorlinks=true,linkcolor=blue]{hyperref}
\usepackage[hyperref,dvipsnames]{xcolor}

\usepackage{cite}

\usepackage{amsthm}
\usepackage{amsfonts,amsmath,amssymb}
\usepackage{lmodern}
\usepackage[T1]{fontenc}
\usepackage[utf8]{inputenc}
\usepackage{xurl}

\usepackage[normalem]{ulem}
\usepackage{array}
\usepackage{pgfplots}

\usepackage[capitalise,noabbrev]{cleveref}
\usepackage{bm}
\usepackage{thmtools}
\usepackage{thm-restate}

\usepackage{caption}
\usepackage{subcaption}

\providecolor{DarkBlue}{rgb}{0,0,.545}
\providecolor{DarkGreen}{rgb}{0,.392,0}
\hypersetup{citecolor=DarkGreen}
\hypersetup{linkcolor=DarkBlue}
\hypersetup{urlcolor=DarkBlue}
\usepackage[final]{microtype}
\usepackage{booktabs,multirow}
\usepackage{algorithm,algorithmicx}
\usepackage[noend]{algpseudocode}
\usepackage{xcolor}
\usepackage{bm}
\usepackage{placeins}
\usepackage{caption}
\usepackage[hang]{footmisc}
\usepackage{nicefrac}

\usetikzlibrary{patterns,decorations.pathreplacing,decorations.markings,arrows}
\definecolor{darkred}{rgb}{0.5,0,0}
\definecolor{darkblue}{rgb}{0,0,0.5}
\definecolor{darkgreen}{rgb}{0,0.5,0}

\mathchardef\mhyphen="2D
\def\makeuppercase#1{
	\expandafter\newcommand\csname cal#1\endcsname{\mathcal{#1}}
	\expandafter\newcommand\csname bf#1\endcsname{\textbf{#1}}
}
\def\makelowercase#1{
	\expandafter\newcommand\csname bf#1\endcsname{\textbf{#1}}
}

\newcounter{char}
\loop
\edef\letter{\alph{char}}
\edef\Letter{\Alph{char}}
\expandafter\makelowercase\letter
\expandafter\makeuppercase\Letter
\stepcounter{char}
\unless\ifnum\thechar>26
\repeat

\newcommand{\linesref}[1]{line~\ref{#1}}

\algblockdefx{MRepeat}{EndRepeat}{\textbf{repeat} }{}
\algnotext{EndRepeat}

\def\makeuppercase#1{
	\expandafter\newcommand\csname cal#1\endcsname{\mathcal{#1}}
	\expandafter\newcommand\csname bf#1\endcsname{\textbf{#1}}
}

\newcommand{\E}{\mathbb{E}}
\newcommand{\F}{\mathbb{F}}
\newcommand{\N}{\mathbb{N}}

\newcommand{\p}[1]{\Pr\left[#1\right]}

\renewcommand{\vec}[1]{\mathbf{#1}}
\newcommand{\wt}{\mathrm{wt}}

\newcommand{\tmo}[1]{\tilde{\mathcal{O}}\left({#1}\right)}
\newcommand{\tmt}[1]{\tilde{\Theta}\left({#1}\right)}
\newcommand{\tmom}[1]{\tilde{\Omega}\left({#1}\right)}
\newcommand{\NNSolver}{\textsc{Closest-Pair}}
\newcommand{\NNP}{\mathcal{CP}_{d,\lambda,\gamma}}
\newlength{\strutdepth}%
\newcommand{\h}{{H}}
\newcommand{\hi}{{H^{-1}}}

\newcommand{\block}[3]{(#1)_{B_{#2,#3}}}

\newcommand{\prb}{p}
\newcommand{\prw}{q}

\newcommand{\notes}[3]{
	\noindent{%
		\color{#1}{[#3]}\color{#1}}%
	\strut\vadjust{\kern-\strutdepth%
		\vtop to \strutdepth{%
			\baselineskip\strutdepth%
			\vss\llap{{\large\color{#1}\textbf{#2}\quad\color{black}}}\null%
		}%
	}%
}

\newtheorem{lemma}{Lemma}
\newtheorem{corollary}{Corollary}
\newtheorem{definition}{Definition}
\newtheorem{theorem}{Theorem}
\usetikzlibrary{pgfplots.groupplots}
\usepgfplotslibrary{fillbetween}
\pgfplotsset{compat=1.3}

\title{A Faster Algorithm for Finding Closest Pairs in Hamming Metric}

\author[1]{Andre Esser}
\author[2]{Robert Kübler}
\author[3]{Floyd Zweydinger}

\affil[1]{Cryptography Research Center, Technology Innovation Institute, Abu Dhabi, UAE, \texttt{andre.esser@tii.ae}}
\affil[2]{Medion AG Essen, Germany, \texttt{robert.kuebler@rub.de}}
\affil[3]{Ruhr University Bochum, Germany, \texttt{floyd.zweydinger@rub.de}}

\date{}
\begin{document}

\maketitle

\begin{abstract}
We study the Closest Pair Problem in Hamming metric, which asks to find the pair with the smallest Hamming distance in a collection of binary vectors. We give a new randomized algorithm for the problem on uniformly random input outperforming previous approaches whenever the dimension of input points is small compared to the dataset size. For moderate to large dimensions, our algorithm matches the time complexity of the previously best-known locality sensitive hashing based algorithms. Technically our algorithm follows similar design principles as Dubiner (IEEE Trans. Inf. Theory 2010) and May-Ozerov (Eurocrypt 2015). Besides improving the time complexity in the aforementioned areas, we significantly simplify the analysis of these previous works. We give a modular analysis, which allows us to investigate the performance of the algorithm also on non-uniform input distributions. Furthermore, we give a proof of concept implementation of our algorithm which performs well in comparison to a quadratic search baseline. This is the first step towards answering an open question raised by May and Ozerov regarding the practicability of algorithms following these design principles. 
\end{abstract}

\section{Introduction}
Finding closest pairs in a given dataset of binary points is a fundamental problem in theoretical computer sciences with numerous applications in data science, machine learning, computer vision, cryptography, and many others. 

Image data for example is often represented via compact binary codes to allow for efficient closest pair search in applications like similarity search in images or facial recognition systems \cite{calonder2010brief,lu2015learning,strecha2011ldahash}. The usage of binary codes also allows decoding the represented data to common codewords. Here, the most efficient algorithms known for decoding such random binary linear codes also heavily benefit from improved  algorithms for the Closest Pair Problem \cite{EC:MayOze15,PQCRYPTO:BotMay18}.  Another common application lies in the field of bioinformatics, where the analysis of genomes involves closest pair search on large datasets to identify the most correlated genetic markers \cite{marchini2005genome,musani2007detection}.

To be more precise, the Closest Pair Problem asks to find the pair of vectors with the minimal Hamming distance among $n$ given binary vectors. While the general version of this problem does not make any restrictions on the distribution of input points, several settings imply a uniform distribution of dataset elements \cite{PQCRYPTO:BotMay18,marchini2005genome,musani2007detection,EC:MayOze15}. Usually, in such settings, there is a planted pair, which attains relative distance $\gamma\in[0,\frac{1}{2}]$, which has to be found. This uniform version is also known as the \emph{light bulb problem}\cite{valiant1988functionality}.
The problem can be solved in time linearly in the dataset size\footnote{here we ignore polylogarithmic factors in the dataset size} as long as the dimension of vectors is constant \cite{bentley1980multidimensional,khuller1995simple}. As soon as the dimension is non-constant an effect occurs known as \emph{curse of dimensionality}, which lets the problem become much harder. 

The most common framework to assess the problem is based on \emph{locality-sensitive hashing} (LSH), whose research was initiated in the pioneering work of Indyk and Motwani \cite{STOC:IndMot98}. Roughly speaking, a locality-sensitive hash function is more likely to hash points that are close to each other to the same value, rather than points that are far apart. To solve the Closest Pair Problem leveraging an LSH family one chooses a random hash function of the family and computes the hash value of all points in the dataset. In a next step, one computes the pairwise distance only for those pairs, hashing to the same value. This process is then repeated for different hash functions until the closest pair is found. The initial algorithm by Indyk-Motwani achieves a time complexity of $n^{\log_2(\frac{2}{1-\gamma})}$. In general, a time lower bound of $n^{\frac{1}{1-\gamma}}$ is known for LSH based algorithms \cite{dubiner2010bucketing,motwani2006lower}. In \cite{dubiner2010bucketing} Dubiner also gives an abstract idea of an algorithm achieving this lower bound.\footnote{A precise algorithmic description and a proof of the running time in the case where the vector length is restricted (referred to as \emph{limited amount of data case} in his work) is missing.} Later May and Ozerov \cite{EC:MayOze15} gave the first concrete algorithmic description following similar design principles, also achieving the mentioned lower bound. Additionally, current data-dependent hashing schemes \cite{andoni2015optimal}, where the hash function depends also on the actual points in the dataset, improve on the initial idea by Indyk-Motwani and also match the time lower bound of  \cite{dubiner2010bucketing,motwani2006lower}.

In the uniform setting Valiant \cite{valiant2012finding} was able to circumvent the lower bound by leveraging fast matrix multiplication and hence breaking out of the LSH framework to give an algorithm that runs in time  $n^{1.63}\mathrm{poly}(d)$. Remarkably, the complexity exponent of Valiant's algorithm does not depend on the relative distance $\gamma$ at all. Later this bound was improved to $n^{1.58}\mathrm{poly}(d)$ by Karpa et al. \cite{SODA:KarKasKoh16} and simplified in an elegant algorithm by Alman \cite{DBLP:conf/soda/Alman19} achieving the same complexity.
\medskip 

All mentioned algorithms have in common, that they assume a dimension of $d=c(n)\log(n)$, where $c(n)$ is at least a big constant. The explicit size of those constants is usually not stated, instead an asymptotic argument yields their existence. Moreover,  the results by \cite{andoni2015optimal, dubiner2010bucketing, valiant2012finding} for instance assume even larger dimensions where $\frac{1}{c(n)}=o(1)$ . Here, the algorithm by May-Ozerov forms an exception by being applicable for any $c(n)\geq \frac{1}{1-H(\frac{\gamma}{2})}$, where $H(\cdot)$ denotes the binary entropy function. Nevertheless, the mentioned lower bound is only achieved for $c(n)$ approaching infinity.
Recently, Xie, Xu and Xu \cite{xie2019new} proposed a new algorithm based on decoding the points of the data set according to some random code, exploiting that close vectors are more likely to be decoded to the same word. Their algorithm is also applicable for any $c(n)$ that allows to bound the number of pairs attaining relative distance $\gamma$ to a constant number with high probability. The authors are able to derandomize their approach and, thus, obtain the fastest known deterministic algorithm for small constants $c(n)$. However, if one also considers probabilistic procedures, their method is inferior to the one by May-Ozerov.

\subsection{Our Contribution}

We design a randomized algorithm, which achieves the best-known running time for solving the Closest Pair Problem on uniformly random input, when the dimension $d$ is small, which means $c(n)$ being a small constant. Additionally, our algorithm matches the running time of the best known LSH algorithms for larger values of $c(n)$ and still matches the time lower bound for LSH based schemes if $\frac{1}{c(n)}=o(1)$. To quantify we give in \cref{fig:cmp-mo-new} the achieved runtime exponent for $c(n)\in\{1.2,2,4\}$ of our algorithm in comparison to May-Ozerov. As indicated by the graphics, our algorithm can be seen as a natural extension of the May-Ozerov algorithm to higher distances. Moreover, we show that for large distances our algorithm is indeed optimal. Note that apart from the May-Ozerov algorithm \emph{none} of the previously mentioned algorithms is applicable for those choices of $c(n)$. A detailed comparison to the result of May and Ozerov is given right after \cref{thm:main}. 

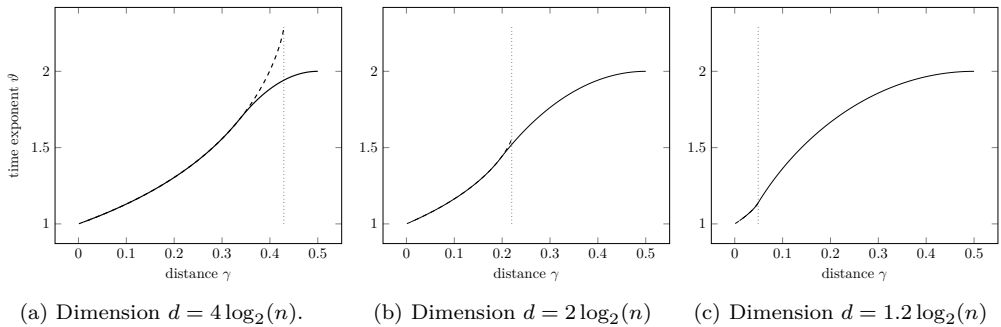
\begin{figure}[t]
	\begin{subfigure}{.3\textwidth}
		
		\begin{tikzpicture}[scale=0.55]
			\begin{axis}[
				y tick label style={
					/pgf/number format/.cd,
					fixed,
					precision=3,
					/tikz/.cd
				},
				x tick label style={
					/pgf/number format/.cd,
					fixed,
					1000 sep={},
					precision=2,
					/tikz/.cd
				},
				xlabel={distance $\gamma$},
				ylabel={time exponent $\vartheta$},
				legend cell align={left},
				]
				
				\addplot [domain=1:2.29, samples=500,mark=none,black,	mark size=1.8, dotted, opacity=0.8] (0.429,x)  node[pos=0.5,yshift=4pt,opacity=0.8,xshift=-55pt] {};
				
				\pgfplotstableread{plots/cmp-new-025}
				\cTradeOff
				\addplot[color=black, thick] table[x = W,y=T] from \cTradeOff ;
				
				\pgfplotstableread{plots/cmp-mo-025}
				\cTradeOff
				\addplot[color=black,dashed, thick] table[x = W,y=T] from \cTradeOff ;
				
			\end{axis}	
			
		\end{tikzpicture}
		\caption{Dimension $d=4\log_2(n)$.}
		\label{fig:cmp-025}
	\end{subfigure}\hspace{0.3cm}
	\begin{subfigure}{.3\textwidth}
		
		\begin{tikzpicture}[scale=0.55]
			\begin{axis}[
				y tick label style={
					/pgf/number format/.cd,
					fixed,
					precision=3,
					/tikz/.cd
				},
				x tick label style={
					/pgf/number format/.cd,
					fixed,
					1000 sep={},
					precision=2,
					/tikz/.cd
				},
				xlabel={distance $\gamma$},
				ylabel={},
				legend cell align={left},
				]
				
			\addplot [domain=1:2.29, samples=500,mark=none,black,	mark size=1.8, dotted, opacity=0.8] (0.22,x)  node[pos=0.5,yshift=4pt,opacity=0.8,xshift=-55pt] {};
			
			\pgfplotstableread{plots/cmp-new-05}
			\cTradeOff
			\addplot[color=black] table[x = W,y=T] from \cTradeOff ;
			
			\pgfplotstableread{plots/cmp-mo-05}
			\cTradeOff
			\addplot[color=black,dashed] table[x = W,y=T] from \cTradeOff ;

			\end{axis}	
			
		\end{tikzpicture}
		
		\caption{Dimension $d=2\log_2(n)$}
		\label{fig:cmp-05}
	\end{subfigure}\hspace{0cm}
\begin{subfigure}{0.3\textwidth}
	
	\begin{tikzpicture}[scale=0.55]
		\begin{axis}[
			y tick label style={
				/pgf/number format/.cd,
				fixed,
				precision=3,
				/tikz/.cd
			},
			x tick label style={
				/pgf/number format/.cd,
				fixed,
				1000 sep={},
				precision=2,
				/tikz/.cd
			},
			xlabel={distance $\gamma$},
			ylabel={},
			legend cell align={left},
			]
			
			\addplot [domain=1:2.29, samples=500,mark=none,black,	mark size=1.8, dotted, opacity=0.8] (0.049,x)  node[pos=0.5,yshift=4pt,opacity=0.8,xshift=-55pt] {};
			
			\pgfplotstableread{plots/cmp-new-083}
			\cTradeOff
			\addplot[color=black] table[x = W,y=T] from \cTradeOff ;
			
			\pgfplotstableread{plots/cmp-mo-083}
			\cTradeOff
			\addplot[color=black,dashed] table[x = W,y=T] from \cTradeOff ;

		\end{axis}	
		
	\end{tikzpicture}
	
	\caption{Dimension $d=1.2\log_2(n)$}
	\label{fig:cmp-083}
\end{subfigure}
	\caption{Time complexity exponent $\vartheta$ as a function of the relative distance $\gamma$ of the closest pair for different dimensions. The running time is of the form $n^\vartheta\cdot \textrm{poly}(d)$, where the dashed line represents May-Ozerov's algorithm and the solid line depicts the exponent of our new algorithm. The dotted line gives the maximal $\gamma$ for which the algorithm by May-Ozerov is still applicable.}
	\label{fig:cmp-mo-new}
\end{figure}

Our improvements over previous work lie in the high density regime, which implies multiple solutions to the Closest Vector Problem. Since the distance alone does not allow to distinguish the planted pair in such cases at least a non-negligible fraction of those pairs needs to be reported, to find the planted pair. The relevance of this setting is mostly given by cryptographic \cite{EC:MayOze15,SODA:BDGL16} and coding-theoretic \cite{hirose2016may,gueye2017generalization,both2018decoding} applications, precisely the decoding of linear codes. Here, the searched error-vector has known weight and is usually constructed in a tree-wise meet-in-the-middle fashion. Even though the error-vector is usually unique the tree-wise decomposition of the problem introduces multiple solution candidates, such that the lists in the tree  can even hold exponentially many pairs with relative distance smaller than $\gamma$. However, in such settings, only the elements attaining relative distance $\gamma$  can possibly sum to the searched error-vector. 
In the algorithm of Both and May \cite{both2018decoding}, which is the fastest known for decoding random binary linear codes, the authors had to define naive fallback routines for the high density case, for which the May-Ozerov algorithm is not applicable. Here our result allows at least for a unified analysis of the algorithm without the need of fallback routines and in the best case leads to runtime improvements.
Also, the generalization of the May-Ozerov nearest neighbor algorithm to $\F_q$ by Hirose \cite{hirose2016may} suffers similar limitations regarding the high density regime, while also forming the basis for the fastest known decoding algorithm for random linear codes over $\F_q$ \cite{gueye2017generalization}. 

Technically our algorithm follows similar design principles as \cite{dubiner2010bucketing,EC:MayOze15}. At its core, these algorithms group the elements of the given datasets recursively into buckets according to some criterion, which fulfills properties that are similar to those of locality-sensitive hash functions. As the buckets in the recursion are decreasing in size, at the end of the recursion they become small enough to compute the pairwise distance of all contained elements naively.

In contrast to previous works, we exchange the used bucket criteria, which allows us to significantly simplify the algorithms' analysis as well as improve for the mentioned parameter regimes. Also, our approach is applicable for any $c(n)$, thus we are able to remove the restriction $c(n)\geq \frac{1}{1-H(\frac{\gamma}{2})}$.

Following May-Ozerov and Dubiner, we study the bichromatic version of the Closest Pair Problem, which takes as input two datasets rather than one and the goal is to find the closest pair between those given datasets. Obviously, there exists a randomized reduction between the Closest Pair Problem and its bichromatic version, but our algorithm can also be easily adapted to the single dataset case. However, May and Ozerov require the elements within each dataset to be pairwise independent of each other, as a minor contribution we get rid of this restriction, too.

Also, we investigate the algorithms' performance on different input distributions. Therefore we give a modular analysis, which allows for an easy exchange of dataset distribution as well as the choice of bucketing criterion. We also give numerical upper bounds for the algorithm's complexity exponent on some exemplary input distributions. These examples suggest that the chosen criterion is well suited as long as the distance between input elements concentrates around $\frac{d}{2}$ (as in the case of random input lists), while being non-optimal as soon as the expected distance decreases.

We also address an open research question regarding the practical applicability of algorithms following the design of \cite{dubiner2010bucketing,EC:MayOze15} raised by May and Ozerov.
As their algorithm inherits a huge polynomial overhead in time and space, they left it as an open problem to give a more practical algorithm following a similar design. While our analysis first suggests an equally high overhead, we are able to give an efficient implementation of our algorithm, which requires in addition to the input dataset only constant space. Also, our practical experiments show that most of the overhead of our algorithm is an artifact of the analysis and can be circumvented in practice so that our algorithm performs well compared to a quadratic search baseline.


\section{Preliminaries}
\subsection{Notation}

For $a,b\in\N$, $a\le b$ we denote $\left[a, b\right]:=\{a, a+1, \ldots, b-1, b\}$. In particular, let $[b]:=[1,b]$. For a vector $\vec{v}\in\F_2^d$ and $I\in [d]$ let $\vec{v}_I$ be the projection of $\vec{v}$ onto the coordinates indexed by $I$, i.e. for $\vec v=(v_1, v_2, \ldots, v_d)$ and $I=\{i_1, i_2, \ldots, i_k\}$ we have $\vec{v}_I=(v_{i_1}, \ldots, v_{i_k})\in\F_2^k$. We denote the uniform distribution on $\F_2^d$ as $\mathcal{U}\left(\F_2 ^d\right)$. We define $f(n)=\tmo{g(n)} :\Leftrightarrow \exists i\in\mathbb{N}\colon f(n)=\mathcal{O}\left(g(n)\cdot \log^i(g(n))\right)$, i.e. the tilde additionally suppresses polylogarithmic factors in comparison to the standard Landau notation $\mathcal{O}$.

Furthermore, we consider all logarithms having base 2. Define the binary entropy function as $\h(x)=-x\log (x) - (1-x)\log (1-x)$ for $x\in(0,1)$, and additionally  $\h(0)=\h(1):=0$. Using this together with Stirling's formula $n!=\Theta\left(\sqrt{2\pi n}\left(\frac{n}{e}\right)^n\right)$ we obtain $\binom{n}{\gamma n}=\tilde{\Theta}\left(2^{\h\left(\gamma\right)n}\right)$. We additionally define $\hi\colon[0, 1]\rightarrow[0, \frac{1}{2}]$ to be the inverse of the left branch of $\h$.

\subsection{Closest Pair Definition}

In this work, we consider the Bichromatic Closest Pair Problem in Hamming metric. Here, the inputs are two lists of equal size containing elements drawn uniformly at random from $\F_2^d$ plus a planted pair, whose Hamming distance is $\gamma d$ for some known $\gamma$. More formally, we state the problem in the following definition. To allow for easy comparison to the result of May-Ozerov, we follow their notation using the dimension as the primary difficulty parameter. Thus we let the list sizes be $n:=2^{\lambda d}$, which means $\lambda = \frac{1}{c(n)}$, where $d=c(n)\log n$.

\begin{definition}[Bichromatic Closest Pair Problem]
	Let $d\in\N$, $\gamma \in \left[0,\frac{1}{2}\right]$ and $\lambda \in(0,1]$. Let $L_1=(\vec v_i)_{i\in [2^{{\lambda d}}]},L_2=(\vec w_i)_{i\in [2^{{\lambda d}}]}\in\left(\F_2^d\right)^{2^{\lambda d}}$ be two lists containing elements uniformly drawn at random, together with a distinguished pair $(\vec{x},\vec{y})\in L_1\times L_2$ with $\wt(\vec{x}+\vec{y})=\gamma d$. We further assume that for each $i,j$ the vectors $\vec v_i$ and $\vec w_j$ are pairwise stochastically independent. The \emph{Closest Pair Problem} $\NNP$ asks to find this \emph{closest pair }$(\vec{x},\vec{y})$ given $L_1,L_2$ and the weight parameter $\gamma$. We call $(\vec x, \vec y)$ the \emph{solution} of the $\NNP$ problem.
	\label{def:NNS}
\end{definition}

First, note that $\lambda\leq1$ is not a real restriction since for $\lambda>1$ the lists must contain duplicates, which can be safely removed, giving us a problem instance with $\lambda\le 1$.
We also consider the Closest Pair Problem on input lists whose elements are distributed according to some distribution $\calD$ different from the uniform one used in \cref{def:NNS}. To indicate this, we refer to the  \emph{$\NNP$ over distribution $\calD$}. Note that in this case, the meaningful upper bound for $\lambda$ is the entropy of $\calD$.

Technically speaking, it is also not necessary to know the value of $\gamma$, as the time complexity of appropriate algorithms to solve the $\NNP$ problem is solely increasing in $\gamma$. Thus if $\gamma$ is unknown, one would apply the algorithm for each $\gamma d=0,1,2,\ldots$ until the solution is found, which results at most in an overhead polynomial in $d$.
 
It is well known, that any LSH based algorithm solving the problem of \cref{def:NNS} with non-negligible probability needs at least time complexity $|L_1|^{\frac{1}{1-\gamma}}=2^{{\frac{\lambda d}{1-\gamma}}}$ \cite{dubiner2010bucketing,motwani2006lower}. However, this lower bound assumes the promised pair to be uniquely distinguishable from all other pairs in $L_1\times L_2$. Obviously, if the relation of $\gamma$ and $\lambda$ lets us expect more than the promised pair of distance $\gamma d$ in the input lists, an algorithm solving the Closest Pair Problem needs to find all (or at least a non-negligible fraction) of these closest pairs.\footnote{Note that in such a scenario the searched $(\vec{x},\vec{y})$ is probably not the pair with the smallest Hamming distance, however, we still refer to elements attaining Hamming distance $\gamma d$ as \emph{closest pairs}.} Hence, if the input lists contain $E$ closest pairs the time complexity of any algorithm solving the problem is lower bounded by
\begin{align*}
\tilde{\Omega}\left(\max(2^\frac{\lambda d}{1-\gamma},\;E)\right)
\end{align*}

Let $(\vec{v}, \vec{w})\in L_1\times L_2\setminus \{(\vec{x}, \vec{y})\}$ be arbitrary list elements. If the elements are chosen independently and uniformly at random, as stated in \cref{def:NNS} we expect $E$ to be of size
\begin{align*}
	\E[E]&=(|L_1\times L_2|-1)\cdot \p{\wt(\vec{v}+\vec{w})=\gamma d} + \underbrace{1}_{\textrm{from } (\vec{x}, \vec{y})}\\
	&=\left(2^{2\lambda d}-1\right)\cdot\frac{\binom{d}{\gamma d}}{2^d}+1\\
		&= \tmt{2^{(2\lambda+\h(\gamma)-1)d}}\enspace,	
\end{align*}

and, thus, the time complexity to solve the $\NNP$ is lower bounded by
\begin{align}
	T_{\textrm{opt}}=\tmom{\max\left(2^{\frac{\lambda d}{1-\gamma}},\;2^{(2\lambda+\h(\gamma)-1)d}\right)}\enspace.
		\label{eqn:Topt-uniform}
\end{align}
\section{Our new Algorithm}
\label{sec:new-alg}
%

Our algorithm groups the input elements according to some criterion into several buckets, each one representing a new closest pair instance with smaller list size. We then apply this bucketing procedure recursively until the buckets contain few enough elements to eventually solve the Closest Pair Problem represented by them via a naive quadratic search algorithm, the exhaustive search. 

As a bucketing criterion, we choose the weight of the vectors after adding a randomly drawn vector $\vec z$ from $\F_2^d$. Thus, each bucket is represented by a vector $\vec z$ and only those elements $\vec v$ are added to the bucket, which satisfy $\wt(\vec v+\vec z)=\delta d$, where $\delta$ is determined later.

\begin{figure}
	\centering
	\usetikzlibrary{patterns,decorations.pathreplacing,arrows.meta}

\begin{tikzpicture}

\pattern[pattern=north west lines,opacity=0.7] (-1.5,9.1) -- (-1.2,9.1) -- (-1.2,9.7) -- (-1.5,9.7) --cycle;
\draw (-1.5,9.1) -- (-1.5,9.7)  -- (-1.2,9.7)  -- (-1.2,9.1) -- (-1.5,9.1); 
\node[align=flush left,text width=4cm] at (1,9.4) {\scriptsize corresponds to\\relative weight $ \delta$ };    

\draw [draw](0,5) -- (-1.5,5) -- (-1.5,7.5) -- (0,7.5) -- (0,5);
\draw[opacity=0.5] (-1.2,7.5) --  (-1.2,5);
\draw[opacity=0.5] (-0.9,7.5) -- (-0.9,5);
\draw[opacity=0.5] (-0.3,7.5) -- (-0.3,5);

\draw [line width=1pt](-1.5,5.4) -- (0,5.4);
\node at (-0.7,5.6) {$\mathbf x$};


\draw [decorate,decoration={brace,amplitude=3pt},xshift=0pt,yshift=0pt]
(-1.5,7.6) -- (-1.2,7.6) node [black,midway,yshift=+0.3cm] 
{\footnotesize $k = \frac{d}{r}$};

\draw [decorate,decoration={brace,amplitude=3pt},xshift=0pt,yshift=0pt]
(-1.6,5) -- (-1.6,7.5) node [black,midway,xshift=-0.4cm] 
{\footnotesize $2^{\lambda d}$};

\draw [decorate,decoration={brace,amplitude=3pt},xshift=0pt,yshift=0pt]
(0,4.9) -- (-1.5,4.9) node [black,midway,yshift=-0.3cm] 
{\footnotesize $d$};

\path (-0.2,6.2) -- node[auto=false]{\dots} (-0.9,6.2);
\path (1.5,6.2) -- node[auto=false]{\dots} (0.8,6.2);

\draw [draw](1.7,5) -- (0.2,5) -- (0.2,7.5) -- (1.7,7.5) -- (1.7,5);
\draw[opacity=0.5] (0.5,7.5) -- (0.5,5);
\draw[opacity=0.5] (0.8,7.5) -- (0.8,5);
\draw[opacity=0.5] (1.4,7.5) -- (1.4,5);

\draw [line width=1pt](0.2,6.9) -- (1.7,6.9);
\node at (1,7.1) {$\mathbf y$};

\draw [draw](4.1,8.4) -- (2.6,8.4) -- (2.6,10.4) -- (4.1,10.4) -- (4.1,8.4);
\draw[opacity=0.5] (2.9,10.4) -- (2.9,8.4);
\draw[opacity=0.5] (3.2,10.4) -- (3.2,8.4);
\draw[opacity=0.5] (3.8,10.4) -- (3.8,8.4);
\pattern[pattern=north west lines,opacity=0.7] (2.6,8.4) --(2.9,8.4) -- (2.9,10.4) -- (2.6,10.4) --cycle;
\draw [draw](5.8,8.4) -- (4.3,8.4) -- (4.3,10.4) -- (5.8,10.4) -- (5.8,8.4);
\draw[opacity=0.5] (4.6,10.4) -- (4.6,8.4);
\draw[opacity=0.5] (4.9,10.4) -- (4.9,8.4);
\draw[opacity=0.5] (5.5,10.4) -- (5.5,8.4);
\pattern[pattern=north west lines,opacity=0.7] (4.3,8.4) --(4.6,8.4) -- (4.6,10.4) -- (4.3,10.4) --cycle;

\draw [draw](4.1,6) -- (2.6,6) -- (2.6,8) -- (4.1,8) -- (4.1,6);
\draw[opacity=0.5] (2.9,8) -- (2.9,6);
\draw[opacity=0.5] (3.2,8) -- (3.2,6);
\draw[opacity=0.5] (3.8,8) -- (3.8,6);
\path (3.8,7) -- node[auto=false]{\dots} (3.2,7);
\pattern[pattern=north west lines,opacity=0.7] (2.6,6) --(2.9,6) -- (2.9,8) -- (2.6,8) --cycle;
\draw [draw](5.8,6) -- (4.3,6) -- (4.3,8) -- (5.8,8) -- (5.8,6);
\draw[opacity=0.5] (4.6,8) -- (4.6,6);
\draw[opacity=0.5] (4.9,8) -- (4.9,6);
\draw[opacity=0.5] (5.5,8) -- (5.5,6);
\path (5.5,7) -- node[auto=false]{\dots} (4.9,7);
\pattern[pattern=north west lines,opacity=0.7] (4.3,6) --(4.6,6) -- (4.6,8) -- (4.3,8) --cycle;

\draw [draw](4.1,2.8) -- (2.6,2.8) -- (2.6,4.8) -- (4.1,4.8) -- (4.1,2.8);
\draw[opacity=0.5] (2.9,4.8) -- (2.9,2.8);
\draw[opacity=0.5] (3.2,4.8) -- (3.2,2.8);
\draw[opacity=0.5] (3.8,4.8) -- (3.8,2.8);
\path (3.8,3.8) -- node[auto=false]{\dots} (3.2,3.8);
\pattern[pattern=north west lines,opacity=0.7] (2.6,2.8) --(2.9,2.8) -- (2.9,4.8) -- (2.6,4.8) --cycle;
\draw [draw](5.8,2.8) -- (4.3,2.8) -- (4.3,4.8) -- (5.8,4.8) -- (5.8,2.8);
\draw[opacity=0.5] (4.6,4.8) -- (4.6,2.8);
\draw[opacity=0.5] (4.9,4.8) -- (4.9,2.8);
\draw[opacity=0.5] (5.5,4.8) -- (5.5,2.8);
\path(5.5,3.8) -- node[auto=false]{\dots} (4.9,3.8);
\pattern[pattern=north west lines,opacity=0.7] (4.3,2.8) --(4.6,2.8) -- (4.6,4.8) -- (4.3,4.8) --cycle;

\draw [-latex](1.7,6.4) -- (2.6,9.2) node[sloped, midway, above, xshift=0.4cm] {${\scriptstyle {\scriptscriptstyle +} \mathbf{z}^{(1)}_1 }$};
\draw [-latex](1.7,6.4) -- (2.6,7) node[sloped, midway, above, xshift=0.3cm] {${\scriptstyle {\scriptscriptstyle +} \mathbf{z}^{(1)}_2 }$};
\draw [-latex](1.7,6.4) -- (2.6,4.1) node[sloped, midway, above, xshift=-0.2cm] {${\scriptstyle {\scriptscriptstyle +} \mathbf{z}^{(1)}_N }$};

\path (4.2,5.9) -- node[auto=false]{\vdots} (4.2,4.9);

\draw [line width=1pt](2.6,8.8) -- (4.1,8.8);
\node at (3.4,9) {$\mathbf x$};
\draw [line width=1pt](4.3,9.6) -- (5.8,9.6);
\node at (5.1,9.8) {$\mathbf y$};

\draw [draw](8.2,10.3) -- (6.7,10.3) -- (6.7,11.4) -- (8.2,11.4) -- (8.2,10.3);
\draw[opacity=0.5] (7,11.4) -- (7,10.3);
\draw [opacity=0.5](7.3,11.4) -- (7.3,10.3);
\draw[opacity=0.5] (7.9,11.4) -- (7.9,10.3);
\path (7.9,10.8) -- node[auto=false]{\dots} (7.3,10.8);
\pattern[pattern=north west lines,opacity=0.7] (7,10.3) --(7.3,10.3) -- (7.3,11.4) -- (7,11.4) --cycle;
\draw [draw](9.9,10.3) -- (8.4,10.3) -- (8.4,11.4) -- (9.9,11.4) -- (9.9,10.3);
\draw[opacity=0.5] (8.7,11.4) -- (8.7,10.3);
\draw[opacity=0.5] (9,11.4) -- (9,10.3);
\draw[opacity=0.5] (9.6,11.4) -- (9.6,10.3);
\path (9.6,10.8) -- node[auto=false]{\dots} (9,10.8);
\pattern[pattern=north west lines,opacity=0.7] (8.7,10.3) --(9,10.3) -- (9,11.4) -- (8.7,11.4) --cycle;

\draw [draw](8.2,9) -- (6.7,9) -- (6.7,10.1) -- (8.2,10.1) -- (8.2,9);
\draw[opacity=0.5] (7,10.1) -- (7,9);
\draw[opacity=0.5] (7.3,10.1) -- (7.3,9);
\draw[opacity=0.5] (7.9,10.1) -- (7.9,9);
\path (7.9,9.5) -- node[auto=false]{\dots} (7.3,9.5);
\pattern[pattern=north west lines,opacity=0.7] (7,9) --(7.3,9) -- (7.3,10.1) -- (7,10.1) --cycle;
\draw [draw](9.9,9) -- (8.4,9) -- (8.4,10.1) -- (9.9,10.1) -- (9.9,9);
\draw[opacity=0.5] (8.7,10.1) -- (8.7,9);
\draw[opacity=0.5] (9,10.1) -- (9,9);
\draw[opacity=0.5] (9.6,10.1) -- (9.6,9);
\path (9.6,9.5) -- node[auto=false]{\dots} (9,9.5);
\pattern[pattern=north west lines,opacity=0.7] (8.7,9) --(9,9) -- (9,10.1) -- (8.7,10.1) --cycle;

\draw [draw](8.2,7.3) -- (6.7,7.3) -- (6.7,8.4) -- (8.2,8.4) -- (8.2,7.3);
\draw[opacity=0.5] (7,8.4) -- (7,7.3);
\draw[opacity=0.5] (7.3,8.4) -- (7.3,7.3);
\draw[opacity=0.5] (7.9,8.4) -- (7.9,7.3);
\pattern[pattern=north west lines,opacity=0.7] (7,7.3) --(7.3,7.3) -- (7.3,8.4) -- (7,8.4) --cycle;
\draw [draw](9.9,7.3) -- (8.4,7.3) -- (8.4,8.4) -- (9.9,8.4) -- (9.9,7.3);
\draw[opacity=0.5] (8.7,8.4) -- (8.7,7.3);
\draw[opacity=0.5] (9,8.4) -- (9,7.3);
\draw[opacity=0.5] (9.6,8.4) -- (9.6,7.3);
\pattern[pattern=north west lines,opacity=0.7] (8.7,7.3) --(9,7.3) -- (9,8.4) -- (8.7,8.4) --cycle;

\draw [-latex](5.8,9) -- (6.7,10.8) node[sloped, midway, above, xshift=0.4cm] {${\scriptstyle {\scriptscriptstyle +} \mathbf{z}^{(2)}_1 }$};
\draw [-latex](5.8,9) -- (6.7,9.7) node[sloped, midway, below, xshift=-0.1cm] {${\scriptstyle {\scriptscriptstyle +} \mathbf{z}^{(2)}_2 }$};
\draw [-latex](5.8,9) -- (6.7,7.9) node[sloped, midway, below, xshift=0.35cm] {${\scriptstyle {\scriptscriptstyle +} \mathbf{z}^{(2)}_N }$};

\path (8.3,9.2) -- node[auto=false]{\vdots} (8.3,8.4);

\draw [line width=1pt](6.7,8) -- (8.2,8);
\node at (7.5,8.2) {$\mathbf x$};
\draw [line width=1pt](8.4,7.6) -- (9.9,7.6);
\node at (9.2,7.8) {$\mathbf y$};


\path (8.4,5.1) -- node[auto=false]{\dots} (7.8,5.1);
\path (10.4,9.3) -- node[auto=false]{\dots} (10.8,9.3);

\end{tikzpicture}
	\caption{We start off on the left side of the illustration with the two input lists $L_1, L_2$ containing the closest pair $(\vec x, \vec y)$. Going right, in each iteration of the algorithm, $N$ different $\vec{z}_i^{(j)}$ are randomly chosen and all of the list elements are tested if they fulfill the bucketing criterion. The crosshatched pattern indicates the parts where the bucket criterion is fulfilled, i.e. the list vectors differ from $\vec{z}_i^{(j)}$ in $\delta k$ positions.}
	\label{fig:tikz:alg:prep}
\end{figure}
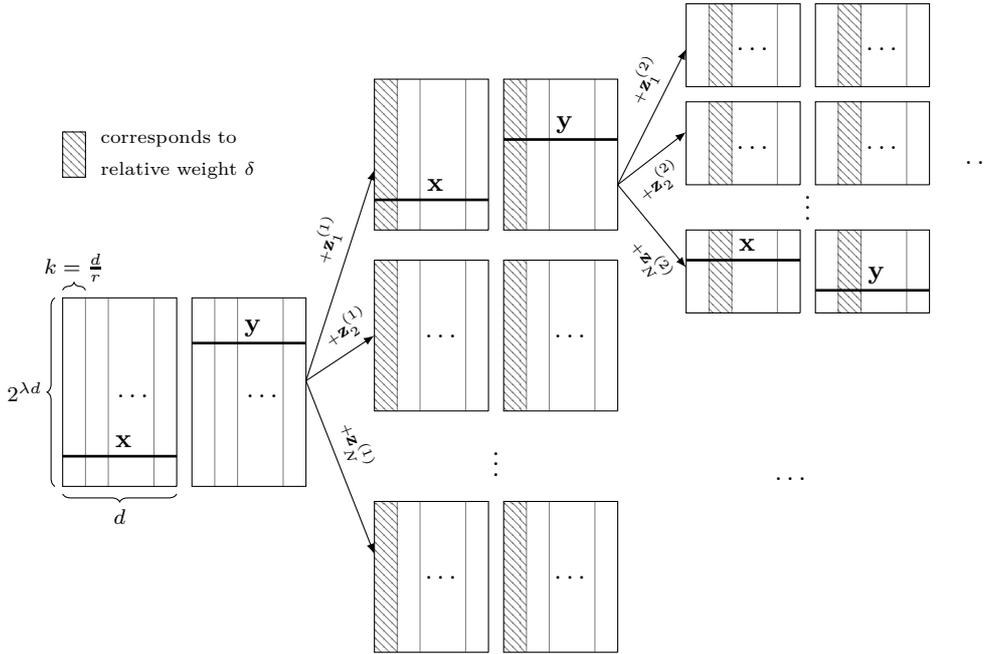

More precisely in each recursive iteration, our algorithm works only on equally large blocks of the input vectors and not on the full $d$ coordinates, i.e. the weight condition is only checked on the current block. This is a technical necessity to obtain independence of vectors in the same bucket on fresh blocks. Let us formally define the notion of blocks.

\newpage
\begin{definition}[Block]
	Let $d, r\in\N$ with $r\mid d$ and $i\in[r]$. Then we denote the $i$-th \emph{block} of $[d]$ as
	\[ B^d_{i, r} := \left[(i-1)\frac{d}{r}+1,\;i\frac{d}{r}\right]\enspace.\]
	Note that $[d] = \bigcup_{i\in [r]}B^d_{i, r}$ and $\left|B^d_{i, r}\right|=\frac{d}{r}$ for each $i\in[r]$. Furthermore, the blocks are disjoint. For a leaner notation and since the role of $d$ does not change in the course of this paper, we omit the index $d$ in the following, thus we write $B_{i,r}:=B^d_{i,r}$.
	\label{def:block}
\end{definition}

Note that May and Ozerov choose the weight of the vectors on random projections as a criterion. In comparison to our variant, their approach involves more parameters and requires extensive re-randomizations of the instance to achieve good success probabilities which together complicates analysis considerably. We cannot rule out the possibility that a different analysis of the May-Ozerov algorithm would allow for an application in the high density regime. However, the complexities of this hypothetical variant are unclear, while our version allows for easy analysis and yields provably optimal complexities in this regime.

In each iteration, we choose the number $N$ of buckets in such a way that with overwhelming probability the closest pair lands in at least one of the buckets. Hence, our algorithm creates a tree with branching factor $N$ with the distinguished pair being contained in one of the leaves. The deeper we get into the tree, the smaller and, hence, the easier the closest pair instances get. An algorithmic description of the whole procedure is given in pseudocode in \cref{alg:new-nn-it}. For convenience, a summary of all parameter choices made in \linesref{line:parameters} of the algorithm can be found in \cref{eqn:all-params} at the end of \cref{sec:new-alg}.

\begin{algorithm}
	\begin{algorithmic}[1]
		\Require{lists $L_1, L_2\in\left(\F_2^d\right)^{2^{\lambda d}}$, weight parameter $\gamma\in\left[0,\frac{1}{2}\right]$}
		\Ensure{list $L$ containing the solution $(\vec x,\vec y)\in L_1\times L_2 $ to the $\NNP$}
		\State Set $r,P,N\in\N,\;\delta\in\left[0, \frac{1}{2}\right]$ properly and define $k:=\frac{d}{r}$\label{line:parameters}
		\For{$P$ permutations $\pi$}\Comment{permutation on the bit positions}
		\State{Stack $S:=[(\pi(L_1), \pi(L_2), 0)]$}
		\State $L\leftarrow\emptyset$
		\While{$S$ is not empty}\label{line:main-loop}
		\State $(A, B, i)\leftarrow S.\textrm{pop}()$
		\If{$i<r$}
		\For{$N$ randomly chosen $\vec z\in\F_2^{k}$}
		\State $A'\leftarrow (\vec v~\in A\mid \wt\big(\block{\vec v~+\vec z}{i+1}{r}\big)=\delta k)$\label{line:bucket-assignment}
		\State $B'\leftarrow (\vec w\in B\mid \wt\big(\block{\vec w+\vec z}{i+1}{r}\big)=\delta k)$\label{line:bucket-assignment2}
		\State $S.\textrm{push}((A', B', i+1))$\label{line:push-new}
		\EndFor
		\Else
		\For{$\vec v \in A, \vec w \in B$} \Comment{Naive search}\label{line:naive-search}
		\If{$\wt(\vec v + \vec w)=\gamma d$}
		\State $L\leftarrow L\cup \{(\vec v,\vec w)\}$
		\EndIf
		\EndFor
		\EndIf
		\EndWhile
		\State\Return $L$
		\EndFor
	\end{algorithmic}
	
	\caption{\NNSolver($L_1, L_2,\gamma$)}
	\label{alg:new-nn-it}
\end{algorithm}

\newpage

The following theorem gives the time complexity of our algorithm to solve the $\NNP$.

\begin{restatable}{theorem}{mainthm}
	Let $\gamma\in\left[0,\frac{1}{2}\right]$ and $\lambda\in[0,1]$. Then \cref{alg:new-nn-it} solves the $\NNP$ problem with overwhelming success  probability in expected time ${2^{\vartheta d(1+o(1))}}$, where
	\[
	\vartheta=\begin{cases} (1-\gamma)\bigg(1-\h\bigg(\frac{\delta^\star-\frac{\gamma}{2}}{1-\gamma}\bigg)\bigg) &\textrm{for } \gamma \leq \gamma^\star \\
		2\lambda+H(\gamma)-1 &\textrm{for } \gamma>\gamma^\star\enspace,
	\end{cases}
	\]
	with  $\delta^\star:=\hi(1-\lambda)$ and $\gamma^\star:=2\delta^\star(1-\delta^\star)$.
	\label{thm:main}
\end{restatable}

Note that the case distinction marks the transition to the high density regime. Precisely, the transition happens when the amount of closest pairs becomes larger than the running time in the first case. In this first case, where $\gamma\leq\gamma^\star$ our algorithm exactly matches the running time of the May-Ozerov algorithm, which itself is shown to match the lower bound for LSH based approaches whenever $\lambda$ approaches zero \cite{EC:MayOze15} (see also \cref{lem:lim_omega}). In the second case, where $\gamma>\gamma^\star$ the running time of our algorithm becomes linear in the number of closest pairs, hence it matches the lower bound from \cref{eqn:Topt-uniform}, while the running time of May-Ozerov stays as in the first case. Our algorithm hence optimally extends the May-Ozerov algorithm to the high density regime.

We establish the proof of \cref{thm:main} in a series of lemmata and theorems. Note that any bucketing algorithm heavily depends on two probabilities specific to the chosen bucketing criterion. First, the probability that any element falls into a bucket, which we call $\prb$ in the remainder of this work. This probability is mainly responsible for the lists' sizes throughout the algorithm. The second relevant probability, which we call $\prw$ describes the event of both, $\vec x$ and $\vec y$, falling into the same bucket, where $(\vec x,\vec y)$ is the solution to the $\NNP$ problem. This is the probability of $(\vec x, \vec y)$ \emph{surviving} one iteration meaning that $\prw$ determines the success probability of the algorithm. In summary, for our choice of bucketing criteria, we get
\begin{align}
	\begin{split}
		\prb &:=\underset{\vec z}{\textrm{Pr}} \left[\wt(\block{\vec v+\vec z}{i}{r})=\delta k\right]\textrm{ for any } \vec v\in\F_2^k \textrm{ and }\\
		\prw&:=\underset{\vec z}{\textrm{Pr}} \left[\wt(\block{\vec{x}+\vec{z}}{i}{r})=\wt(\block{\vec{y}+\vec{z}}{i}{r})=\delta k\right]\label{eqn:probabilities}\enspace,
	\end{split}
\end{align}
where $k=\frac{d}{r}$ is the block width. If we assume that the $\gamma d$ differing coordinates of $\vec x$ and $\vec y$ distribute evenly into the $r$ blocks, i.e. $\wt(\block{\vec{x}+\vec{y}}{i}{r})=\gamma k$ for each $i$, these probabilities are independent of $i$ for $\delta k $ fixed. This property is ensured for at least one of the $P$ permutations in \cref{alg:new-nn-it} with overwhelming probability, as we will see in the proof of \cref{thm:main-easy}.

We determine the exact form of $\prb$ and $\prw$ later. First, we are going to prove the following statement about the expected running time of \cref{alg:new-nn-it} in dependence on both probabilities.

\begin{theorem}\label{thm:main-easy}
	Let $\prb$ and $\prw$ be as defined in \cref{eqn:probabilities}, $\gamma\in\left[0,\frac{1}{2}\right] $, $\lambda\in[0, 1]$ and $r=\frac{\lambda d}{\log^2 d}$. Then \cref{alg:new-nn-it} solves the $\NNP$ problem in expected time
	\[	\max\bigg({\prw^{-r}} ,\frac{2^{\lambda d}\cdot \prb^{r-1}}{\prw^r},\frac{\left(2^{\lambda d}\cdot \prb^{r}\right)^2}{\prw^r}\bigg)^{1+o(1)}\
	\]
	with a success probability overwhelming in $d$. 
	\label{thm:main-qq2}
\end{theorem}

\begin{proof}
	First, we are going to prove the statement about the time complexity. 
	
	The algorithm maintains a stack, containing list pairs together with an associated counter. In every iteration of the loop in \linesref{line:main-loop}, one element is removed from the stack and if the counter $i$ associated with this element is smaller than $r$, $N$ additional elements $(A',B',i+1)$ are pushed to the stack in \linesref{line:push-new}. Let us consider the elements on the stack as nodes in a tree of depth $r$, where all elements with associated counter $i$ are siblings on level $i$ of the tree. Also, depict the elements pushed to the stack in \linesref{line:push-new} as child nodes of the currently processed node $(A,B,i)$. Then the total number of elements with associated counter $i$ pushed to the stack is bounded by the number of nodes on level $i$ in a tree with branching factor $N$, which is $N^i$.
	
	Next, let us determine the lists' sizes on level $i$ of that tree. Therefore, let the expected size of lists on level $i$ be $\calL_i$. As these lists are constructed from the lists of the previous level by testing the weight condition in \linesref{line:bucket-assignment} and \ref{line:bucket-assignment2}, it holds that
	\[\calL_i=\calL_{i-1}\cdot\p{\wt(\block{\vec v + \vec z}{i}{r}))=\delta k}:=\calL_{i-1}\cdot \prb\enspace,
	\]
	where $i>0$ and by construction $\calL_0=|L_1|$. By substitution we get
	\[
	\calL_i=|L_1|\cdot \prb^{i}\enspace, \textrm{ for } i=0,\ldots,r.\] 
	
	Now, we are able to compute the time needed to create the nodes on level $i$ of the tree. Observe that for the creation of a level-$i$ node we need to linearly scan through the larger lists of a node on level $i-1$ to check the weight conditions. Thus, to construct all $N^{i}$ nodes of level $i$ we need a total time of
	\[
		T_{i}=\tmo{\calL_{i-1} \cdot N^{i}}=\tmo{|L_1|\cdot \prb^{i-1}\cdot N^{i}}\enspace,
	\]
	for each $0<i\leq r$. Eventually, the list pairs on level $r$ are matched by a naive search with quadratic runtime resulting in
	\[
		T_{r+1}=\tmo{N^r\cdot \E[|A_r|\cdot |B_r|]}\enspace,
	\]
	where $A_r,B_r$ describe the lists of a level-$r$ node.
	
	The expected value of the product, now, depends on the chosen input distribution. We next argue that for the given input distribution we have 
		\[
	\E[|A_r|\cdot |B_r|]=\calO\big( \E[|A_r|]\cdot\E[|B_r|]\big)=\calO(\calL_r^2)\enspace .
	\]
	
	To see this, first note that for $\vec v, \vec w, \vec z$ independent and uniform, $\vec v + \vec z$ and $\vec w + \vec z$ are also independent and uniform. This in turn implies
	\begin{align*}
		& \p{\wt(\block{\vec v + \vec z}{i}{r})) = \delta k, \wt(\block{\vec w + \vec z}{i}{r}))=\delta k} \\
		=& \p{\wt(\block{\vec v + \vec z}{i}{r}))=\delta k}\cdot \p{\wt(\block{\vec w + \vec z}{i}{r}))=\delta k} \\
		=& \prb^2
	\end{align*}
		since deterministic functions of independent random variables are still independent. This also works for either $\vec v = \vec x$ or $\vec w = \vec y$, but not for $(\vec v, \vec w)=(\vec x, \vec y)$.  In this case, however, we have $\p{\wt(\block{\vec x + \vec z}{i}{r})) = \delta k, \wt(\block{\vec y + \vec z}{i}{r}))=\delta k}=\prw$ by definition. With this insight, we can express $\E[|A_{i}|\cdot |B_{i}|]$ in terms of $\E[|A_{i-1}|\cdot |B_{i-1}|]$ for each $i$ via
	\begin{align*}
		\E[|A_{i}|\cdot|B_{i}|\mid A_{i-1}, B_{i-1}]&=\sum_{\substack{\vec v \in A_{i-1},~\vec w \in B_{i-1}\\ (\vec v, \vec w)\neq(\vec x, \vec y)}}\p{\wt(\block{\vec v + \vec z}{i}{r}) = \delta k, \wt(\block{\vec w + \vec z}{i}{r})=\delta k} \\
		&+\p{\wt(\block{\vec x + \vec z}{i}{r})) = \delta k, \wt(\block{\vec y + \vec z}{i}{r}))=\delta k}\\
		&=(|A_{i-1}|\cdot|B_{i-1}|-1)\prb^2+\prw\\
		&\le |A_{i-1}|\cdot|B_{i-1}|\cdot \prb^2 +1\enspace,
	\end{align*}
Applying the Law of total Expectation we obtain
\begin{align}\E[|A_{i}|\cdot|B_{i}|]=\E[\E[|A_{i}|\cdot|B_{i}|\mid A_{i-1}, B_{i-1}]]\le \E[|A_{i-1}|\cdot|B_{i-1}|]\cdot \prb^2 + 1
\label{eqn:sizes}
\end{align}
Successive application of \cref{eqn:sizes} yields
\begin{align}
\E[|A_{r}|\cdot|B_{r}|] \le \E[|L_1|\cdot|L_2|]\cdot \prb^{2r} + r = 2^{2\lambda d} \prb^{2r}+r=\calO(\calL_r^2)
\label{eqn:expectation-random}
\end{align}

	Finally, the algorithm is repeated for $P$ different permutations on the bit positions of elements in $L_1,L_2$.
	In summary, the expected time complexity to build all list becomes the sum of the $T_i$ multiplied by $P$, thus, by choosing $N:=\frac{d}{\prw}$ and $P=(d+1)^{r+1}$ we get
	
	\begin{align*}
		T'&=P\cdot \sum_{i=1}^{r+1} T_i\leq (d+1)^{r+1}\cdot \left( \sum_{i=1}^{r} N^{i}\cdot |L_1|\cdot \prb^{i-1} + (|L_1|\cdot \prb^r)^2\cdot N^r \right)\\
		&=(d+1)^{r+1}\cdot \left( \sum_{i=1}^{r}\frac{|L_1|\cdot d^{i}}{\prw}\cdot\left(\frac{\prb}{\prw}\right)^{i-1} + \frac{\left(|L_1|\cdot \prb^{r}\right)^2\cdot d^{r}}{\prw^r} \right)\\
		&\leq (d+1)^{2r+1}\cdot \bigg( \frac{r\cdot |L_1|\cdot \prb^{r-1}}{\prw^{r}} + \frac{\left(|L_1|\cdot \prb^{r}\right)^2}{\prw^r}\bigg)\\
		&=\max\bigg(\frac{ 2^{\lambda d}\cdot \prb^{r-1}}{\prw^{r}} ,\frac{\left(2^{\lambda d}\cdot \prb^{r}\right)^2}{\prw^r}\bigg)^{1+o(1)}\enspace ,
	\end{align*}
	where the inequality follows from the fact that $\frac{\prb}{\prw}\geq1$ since
	\begin{align*}
		\prw&=\p{\wt(\block{\vec{x}+\vec{z}}{i}{r})=\wt(\block{\vec{y}+\vec{z}}{i}{r})=\delta k}\\
		&\le\p{\wt(\block{\vec{x}+\vec{z}}{i}{r})=\delta k}\\
		&=\prb\enspace,
	\end{align*}
	and the final equality stems from the fact that $|L_1|=2^{\lambda d}$ and $r=o(\frac{\lambda d}{\log d})$ as given in the theorem.
	
	Note that $T'$ disregards the fact that no matter how small the lists in the tree become, the algorithm needs to traverse all 
	\[
	T''= \tmo{N^r}=\tmo{\left(\frac{d}{\prw}\right)^r}
	\]
	 nodes of the tree. Hence, the expected time complexity of the whole algorithm is
	 \[
	 T=\max(T',T'')\enspace,
	 \]
	 which proves the claim.
	\medskip
	
	Let us now consider the success probability of the algorithm. Therefore, we assume that the chosen permutation distributes the weight on $\vec x+\vec y$ such that in every block of length $r$ the weight is equal to $\frac{\gamma d}{r}$, which we describe as a \emph{good} permutation. The probability of a random permutation $\pi$ distributing the weight in such a way is
	\[\p{\textrm{good }\pi} = \p{\wt\left(\pi\block{\vec x + \vec y}{i}{r}\right)=\frac{\gamma d}{r}, \textrm{ for }i=1,\ldots,r} = \frac{\binom{\frac{d}{r}}{\frac{\gamma d}{r}}^r}{\binom{d}{\gamma}}\geq \left(\frac{d}{r}+1\right)^{-r}\enspace.\]

	Thus, the probability of at least one out of $(d+1)^{r+1}$ chosen permutations being good is
	\begin{align*}
	p_1&:=\p{\textrm{at least one good }\pi}\\
	&=1-(1-\p{\textrm{good }\pi})^{(d+1)^{r+1}}
	=1-\left(1-\left(\frac{d}{r}+1\right)^{-r}\right)^{(d+1)^{r+1}}\geq1-e^{-d}\enspace.
	\end{align*}
	 
	 The algorithm succeeds, whenever there exists a leaf node in the tree, containing the distinguished pair $(\vec x, \vec y)$. As every node in the tree is constructed based on its parent, it follows that all nodes on the path from the root to that leaf need to contain $(\vec x, \vec y)$. By definition the probability of $\vec x$ and $\vec y$ satisfying the bucket criterion at the same time (thus for the same $\vec z$) is $\prw$ and since we condition on a good permutation, $\prw$ is equal for every considered block. Let us define indicator variables $X_{j}$ for the first level, where $X_{j}=1$ iff the $j$-th node contains $(\vec x,\vec y)$. Observe that the $X_{j}$ for independent choices of $\vec z$ are independent. Thus, clearly the number of trials until $(\vec x,\vec y)$ is contained in any node on level one is distributed geometrically with parameter $\prw$. Hence, the probability of the solution being contained in at least one node on the first level is
	\begin{align*}
	p_2&:=\p{\exists (A,B,1)\in S : (\vec x,\vec y)\in A\times B }\\
	&=1-(1-\prw)^N=1-(1-\prw)^{\frac{d}{\prw}}\ge 1-e^{-d}\enspace.
	\end{align*} 
	Now, imagine the pair being contained in some level-$i$ node. Considering that node, we have with the same probability $p_2$ again that at least one child contains the solution, and the same argument holds until we reach the leaves. Also, by the independent choices of $\vec z$ the events remain independent which implies that the probability of $(\vec x, \vec y)$ being contained in a level-$r$ list is $p_2^r$. In summary, the success probability is
	\begin{align*}
	\p{\textrm{success}}=p_1\cdot p_2^r&\geq (1-e^{-d})^{r+1} \geq1-\frac{r+1}{e^d}\geq 1-\frac{d}{e^d} \enspace.
	\end{align*}
\end{proof}

The proof of \cref{thm:main-qq2} already shows, how different distributions may affect the complexity of the algorithm by changing the expected value $\E[|A_r|\cdot|B_r|]$. This influence on the algorithms complexity by different input distributions is further investigated in \cref{sec:distributions}.

In the next two lemmata, we will determine the exact forms of $\prb$ and $\prw$ to conduct the run time analysis. 
\begin{lemma}\label{lem:xinT}
	Let $k\in\mathbb{N}$, $\delta\in\left[0, 1\right]$. If $\vec{x}\in\F_2^k$ and $\vec{z}\sim\mathcal{U}(\F_2^k)$ then \[\underset{\vec z}{\mathrm{Pr}}\left[\wt(\vec{x}+ \vec{z})=\delta k\right]=\binom{k}{\delta k}\left(\frac{1}{2}\right)^k.\]
\end{lemma}
\begin{proof}
	Since $\vec{z}\sim\mathcal{U}(\F_2^k)$, the probability is \[\frac{\left|\{\vec{z}\in\F_2^k\mid\wt(\vec{x}+ \vec{z})=\delta k\}\right|}{\left|\F_2^k\right|}.\]
	To compute the numerator, note that $\wt(\vec{x}+ \vec{z})=\delta k$ means that $\vec{x}$ and $\vec{z}$ differ in $\delta k$ out of $k$ coordinates, for which there are $\binom{k}{\delta k}$ possibilities. Using $\left|\F_2^k\right|=2^k$, the lemma follows.
\end{proof}


Before we continue, let us make a small definition.

\begin{definition}
	Let $k\in\mathbb{N}$ and $\vec{x}, \vec{y}\in\F_2^k$. Then we define $D(\vec{x}, \vec{y})\subseteq [k]$ to be the set of coordinates where $\vec{x}$ and $\vec{y}$ differ, i.e.
	\[D(\vec{x}, \vec{y}):=\{i\in[k]\mid \vec{x}_i\not=\vec{y}_i\}.\]
	Furthermore, let $S(\vec{x}, \vec{y}):=[k]\setminus D(\vec{x}, \vec{y})$ be the set of coordinates where they are the same.
\end{definition}

Now we derive the exact form of the probability $\prw$ of a pair with difference $\gamma k$ falling into the same bucket.
\begin{lemma}\label{lem:xyinT}
	Let $k\in\mathbb{N}$, $\delta\in\left[0, 1\right]$. If $\vec{x}, \vec{y}\in\F_2^k$ with $\wt(\vec{x}+ \vec{y})=\gamma k$ and $\vec{z}\sim\mathcal{U}(\F_2^k)$. Then \[\underset{\vec z}{\mathrm{Pr}}\left[\wt(\vec{x}+ \vec{z})= \wt(\vec{y}+ \vec{z})=\delta k\right]=\binom{\gamma k}{\frac{1}{2}\gamma k}\binom{(1-\gamma) k}{\left(\delta - \frac{\gamma}{2}\right)k}\left(\frac{1}{2}\right)^k.\]
\end{lemma}

\begin{proof}
	Let \[A:=\{\vec z\in\F_2^k\mid\wt(\vec{x}+ \vec{z})=\wt(\vec{y}+ \vec{z})=\delta k\}.\] In analogy to Lemma \ref{lem:xinT}, the probability we search for is $\frac{\left|A\right|}{\left|\F_2^k\right|}=\left|A\right|\cdot\left(\frac{1}{2}\right)^k.$
	
	In the following, let $\gamma_{\vec x}:=\wt(\vec x + \vec z)$ and analogously $\gamma_{\vec y}:=\wt(\vec y + \vec z)$. Now observe that every coordinate $z_i$ of $\vec z$ with $i\in S(\vec x,\vec y)$, so belonging to the set of equal coordinates between $\vec x$ and $\vec y$, either contributes to both $\gamma_{\vec x}$ \emph{and} $\gamma_{\vec y}$ with one or does not affect either one of them. Let us define the amount of the $z_i$'s with $i\in S(\vec x,\vec y)$ that contribute to the weight as $a:=|S(\vec x,\vec y)\cap D(\vec x,\vec z)|$.
	
	Now consider the $z_i$'s  with $i\in D(\vec x,\vec y)$. Clearly, any such $z_i$ contributes \emph{either} to $\gamma_{\vec x}$ \emph{or} to $\gamma_{\vec y}$. Thus, let us define the number of those $z_i$ with $i\in D(\vec x,\vec y)$ that contribute to $\gamma_{\vec x}$ as $b_\vec{x}:=|D(\vec x,\vec y)\cap D(\vec x, \vec z)|$ and analogously those which contribute to $\gamma_{\vec y}$ as $b_\vec{y}:=|D(\vec x,\vec y)\cap D(\vec y, \vec z)|$. Obviously we have 
	\begin{align}
		b_{\vec x} + b_{\vec{y}}=|D(\vec x, \vec y)|=\gamma k
	\end{align}
	On the other hand we are only interested in those $\vec z$ for which $\gamma_{\vec x}=\gamma_{\vec y}=\delta k$, which yields the two equations
	\begin{align}
		\gamma_{\vec x}&=a+b_{\vec x} = \delta k \\
		\gamma_{\vec y}&=a+b_{\vec y} = \delta k
	\end{align}
	All three equations together yield the unique solution
	\[
	b_{\vec{x}}=b_{\vec{y}}=\frac{\gamma k}{2} \textrm{ and } a=\left(\delta-\frac{\gamma}{2}\right)k\enspace.
	\]
	This shows the following: If $\vec{z}\in A$, it is necessary that $\vec{z}$ differs from $\vec{x}$ (analogously $\vec y$) in exactly
	\begin{enumerate}
		\item[--] $\frac{\gamma}{2}k$ out of $\gamma k$ coordinates of $D(\vec{x}, \vec{y})$ and
		\item[--] $\left(\delta-\frac{\gamma}{2}\right)k$ out of $(1-\gamma) k$ coordinates of $S(\vec{x}, \vec{y})$.
	\end{enumerate}
	
	Thus, because we can freely combine both conditions, in total there are \[\left|A\right|=\binom{\gamma k}{\frac{\gamma}{2}k}\binom{(1-\gamma) k}{\left(\delta-\frac{\gamma}{2}\right)k}\]
	different values for $\vec{z}$, finishing the proof.
\end{proof}

Now we are ready to prove \cref{thm:main} about the time complexity of \cref{alg:new-nn-it} for solving the $\NNP$ problem. For convenience, we restate the theorem here.

\mainthm*
\begin{proof}
First let us give the exact form of $\log\prb$ and $\log\prw$ using Stirling's formula to approximate the binomial coefficients in \cref{lem:xinT} and \ref{lem:xyinT}. By setting the block width $k=\frac{d}{r}$ we get 
	\begin{align*}
		\log \prw~ &= (1-\gamma) \bigg(H\Big(\frac{\delta-\frac{\gamma}{2}}{1-\gamma}\Big)-1\bigg)\frac{d}{r} \big(1+o(1)\big), \quad
		\log \prb = \big(H(\delta)-1\big)\frac{d}{r}\big(1+o(1)\big)\enspace.
	\end{align*}
	
Now, let us reconsider the running time given in \cref{thm:main-qq2} as 
		\[	T=\max\bigg(\underbrace{{\frac{1}{\prw^r}}}_{(a)},\underbrace{\frac{2^{\lambda d}\cdot \prb^{r-1}}{\prw^r}}_{(b)},\underbrace{\frac{\left(2^{\lambda d}\cdot \prb^{r}\right)^2}{\prw^r}}_{(c)}\bigg)^{1+o(1)}\enspace,
	\]
	where $r=\frac{\lambda d}{\log^2 d}$.
	
	We now show that the running time for all values of $\delta\geq \delta^\star :=\hi(1-\lambda)$ is solely dominated by $(c)$. Observe that we have $(c)\geq(b)$, whenever
	\begin{alignat*}{2}
		&& 2^{\lambda d}\cdot \prb^{2r} &\geq\prb^{r-1}\\
		\Leftrightarrow~&&H(\delta)&\geq 1-\frac{\lambda r}{r+1}\\
		\Leftrightarrow~ &&\delta&\geq \hi\Big(1-\frac{\lambda}{1+\frac{1}{r}}\Big) \rightarrow \hi(1-\lambda) =\delta^\star\enspace, 
		\intertext{	since $\frac{1}{r}=o(1)$. Also we have $(c)\geq(a)$ for the same choice of delta, as }
			&& 2^{2\lambda d}\cdot \prb^{2r} &\geq1\\
			\Leftrightarrow~ &&\delta&\geq \hi(1-\lambda) =\delta^\star\enspace.
		\end{alignat*}
	Thus, for all choices of $\delta\geq\delta^\star$ the running time is $(T_{\delta})^{(1+o(1))}$ with
	\[
	 \vartheta^\star(\delta):=\frac{ \log T_{\delta} }{d}= 2(\lambda+H(\delta)-1)+(1-\gamma)\bigg(1-\h\Big(\frac{\delta-\frac{\gamma}{2}}{1-\gamma}\Big)\bigg)\enspace .
	\]
	
	Now, minimizing $\vartheta^\star$ yields a global minimum at $\delta_\textrm{min}=\frac{1}{2}(1-\sqrt{1-2\gamma})$ attaining a value of 
	\[
		\vartheta^\star(\delta_\textrm{min})=2\lambda+H(\gamma)-1\enspace.
	\]
	
	As we are restricted to values for $\delta$ which are larger than $\delta^\star$ solving $\delta_\textrm{min}\geq\delta^\star$ for $\gamma$ yields
	\begin{alignat*}{2}
	&&  \delta_\textrm{min}&\geq\delta^\star\\
	\Leftrightarrow~&&\gamma&\geq 2\delta^\star(1-\delta^\star)=\gamma^\star\enspace.
\end{alignat*}

This proves the claim of the theorem whenever $\gamma>\gamma^\star$. For all other values of $\gamma$ we simply choose $\delta=\delta^\star$, which yields
\[
   \vartheta = \vartheta^\star(\delta^\star)=(1-\gamma) \bigg(1-\h\Big(\frac{\delta^\star-\frac{\gamma}{2}}{1-\gamma}\Big)\bigg) \textrm{ for } \gamma\leq\gamma^\star
\]
as claimed.

Now to boost the expected running time $2^{\vartheta d (1+o(1))}$ of the algorithm to actually being obtained with overwhelming probability we use a standard Markov argument. Let $X$ denote the random variable describing the running time of the algorithm. Then the probability that the algorithm needs more time than $2^{\sqrt{d}}E[X]$ to finish is
\[
	\p{X\geq  2^{\sqrt{d}}\cdot E[X]}\leq\frac{E[X]}{2^{\sqrt{d}}\cdot E[X]}=2^{-\sqrt{d}}\enspace,
\]
or equivalently the algorithm finishes in less time than $2^{\sqrt{d}}E[X]=2^{\vartheta d (1+o(1))}$ with overwhelming probability. Also, a standard application of the union bound yields that the intersection of the algorithm finishing within the claimed time and the algorithm having success in finding the solution is still overwhelming.
\end{proof}

The theorem shows that whenever $\gamma >\gamma^*$ our algorithm obtains the optimal time complexity for uniformly random lists as given in \cref{eqn:Topt-uniform}. Additionally, our algorithm reaches the time lower bound for locality-sensitive hashing based algorithms for all values of $\gamma$, whenever the input list sizes are subexponential in the dimension $d$, which is shown in the following lemma.

\begin{lemma}\label{lem:lim_omega}
Let $\gamma\in\left[0,\frac{1}{2}\right] $, and $\vartheta$ as defined in \cref{thm:main}. Then we have 
	\[
		\lim\limits_{\lambda\rightarrow0} \frac{\vartheta}{\lambda} = \frac{1 }{1-\gamma}\enspace.
	\]
\end{lemma}
\begin{proof}
Note that for $\lambda$ converging zero, $\delta^\star=\hi(1-\lambda)$ approaches $\frac{1}{2}$. This implies $\gamma^\star:=2\delta^\star(1-\delta^\star)=\frac{1}{2}$ and hence for all choices of $\gamma$ we have
\[\vartheta=(1-\gamma)\bigg(1-\h\bigg(\frac{\delta-\frac{\gamma}{2}}{1-\gamma}\bigg)\bigg)\enspace.\]
Now, for this choice of $\vartheta$,  May and Ozerov \cite[Corollary 1]{EC:MayOze15} already showed  the statement of this lemma, by applying L'Hopital's rule twice.
\end{proof}

For convenience, we restate all parameter choices of \cref{alg:new-nn-it} for solving the $\NNP$ in the following overview:

\begin{align}
	\begin{split}
		r&=\frac{d}{\log^2 d},\,  P=(d+1)^{r+1}, \,k=\frac{d}{r} \\
		N&=\frac{d}{\prw}, \textrm{ where } \prw=\binom{\gamma k}{\frac{1}{2}\gamma k}\binom{(1-\gamma) k}{\left(\delta - \frac{\gamma}{2}\right)k}\left(\frac{1}{2}\right)^k\\[+3mm]
		\delta&=\begin{cases} \delta^\star &\textrm{for } \gamma \leq 2\delta^\star(1-\delta^\star) \\
			\frac{1}{2}(1-\sqrt{1-2\gamma})& \textrm{else}
		\end{cases}\enspace,\textrm{ with } \delta^\star:=H^{-1}(1-\lambda) 
	\end{split}
	\label{eqn:all-params}
\end{align}

\section{Different Input Distributions}
\label{sec:distributions}
In this section, we show how to adapt the analysis of \cref{alg:new-nn-it} to variable input distributions. Therefore, we first reformulate \cref{thm:main-qq2} in \cref{cor:arbitrary-dist} for the case of considering the $\NNP$ over an arbitrary distribution $\calD$. As already indicated in the proof of \cref{thm:main-qq2}, this reformulation depends on the expected value $\calE$ of the cost of the naive search at the bottom of the computation tree, which is highly influenced by the distribution $\calD$. Then, we show how to compute $\calE$ and how to upper bound it effectively.
Finally, we give upper bounds for the time complexity of the algorithm to solve the $\NNP$ over some generic distributions. These examples suggest that the algorithm is best suited for distributions $\calD$, where the weight of the sum $\vec v+\vec w$ of elements $\vec v,\vec w\sim\calD$ concentrates at $\frac{d}{2}$.\footnote{This behavior seems quite natural as in this case, the solution is most distinguishable from random input pairs.}

Let us start with the reformulation of the theorem.
\begin{corollary}
	\label{cor:arbitrary-dist}
	Let $\calD$ be some distribution over $\F_2^d$, $\prw$ and $\prb$ be as defined in \cref{eqn:probabilities}, $\gamma\in\left[0,\frac{1}{2}\right] $, $\lambda\in[0, 1]$ and $r=\frac{\lambda d}{\log^2 d}$. Also let $\calE=\E[|A|\cdot|B|]$ for $A$ and $B$ in \linesref{line:naive-search} of \cref{alg:new-nn-it} (where the expectation is taken over the distribution of input lists and the random choices of the algorithm). Then \cref{alg:new-nn-it} solves the $\NNP$ problem over $\calD$ in time
	\[	\max\bigg({\prw^{-r}} ,\frac{2^{\lambda d}\cdot \prb^{r-1}}{\prw^r},\frac{\calE}{\prw^r}\bigg)^{1+o(1)}\
	\]
	with success probability overwhelming in $d$. 
	\label{cor:main-qq2}
\end{corollary}
\begin{proof}
	The proof follows along the lines of the proof of \cref{thm:main-qq2}, by observing that $T_{r+1}=N^r\cdot \calE$ and the expected time complexity is again amplified to being obtained with overwhelming probability by using a Markov argument similar to the proof of \cref{thm:main}.
\end{proof}

In the next lemma, we show how to upper bound the value of $\calE$.

\begin{lemma}[Expectation of Naive Search]
	Let $\calD$ be some distribution over $\F_2^d$, $\gamma\in\left[0,\frac{1}{2}\right] $, $\lambda\in[0, 1]$ and $r=\frac{\lambda d}{\log^2 d}$. Also let $\calE=\E[|A|\cdot|B|]$ for $A$ and $B$ in \linesref{line:naive-search} of \cref{alg:new-nn-it} when solving some instance of the $\NNP$ over $\calD$ (where the expectation is taken over the distribution of input lists and the random choices of the algorithm). Then we have
	\[
	\calE \leq 2^{2 \lambda d}\prod_{i=1}^{r}\alpha_i + 4r\cdot2^{\lambda d}\cdot \prb^{r}
	\]
	where $\alpha_i := \underset{\vec v,\vec w\sim\calD}{\mathrm{Pr}}\left[{\wt(\block{\vec v + \vec z}{i}{r}) = \delta k, \wt(\block{\vec w + \vec z}{i}{r})=\delta k}\right]$.
	\label{lem:expectation-exact}
\end{lemma} 
\begin{proof}
 
Similar to the proof of \cref{thm:main-easy}, let us bound $\E[|A_{i}|\cdot |B_{i}|]$ in terms of $\E[|A_{i-1}|\cdot |B_{i-1}|]$, $\E[|A_{i-1}|]$ and $\E[|B_{i-1}|]$ for each $i$.

\begin{align*}
	\E[|A_{i}|\cdot|B_{i}|\mid A_{i-1}, B_{i-1}]&=\sum_{\substack{\vec v \in A_{i-1}\setminus\{\vec x\} \\ \vec w \in B_{i-1}\setminus\{\vec y\}}}\underbrace{\p{\wt(\block{\vec v + \vec z}{i}{r}) = \delta k, \wt(\block{\vec w + \vec z}{i}{r})=\delta k}}_{=:\alpha_i} \\
	&+\sum_{\vec v \in A_{i-1}}\underbrace{\p{\wt(\block{\vec v + \vec z}{i}{r}) = \delta k, \wt(\block{\vec y + \vec z}{i}{r})=\delta k}}_{\le \prb}\\
	&+\sum_{\vec w \in B_{i-1}}\underbrace{\p{\wt(\block{\vec x + \vec z}{i}{r}) = \delta k, \wt(\block{\vec w + \vec z}{i}{r})=\delta k}}_{\le \prb} \\
	&+\p{\wt(\block{\vec x + \vec z}{i}{r})) = \delta k, \wt(\block{\vec y + \vec z}{i}{r}))=\delta k}\\
	&\leq \alpha_i\cdot |A_{i-1}|\cdot|B_{i-1}|+ \prb\cdot(|A_{i-1}|+ |B_{i-1}| + 1) 
\end{align*}
and hence $\E[|A_{i}|\cdot|B_{i}|] \le \alpha_i\cdot \E[|A_{i-1}|\cdot|B_{i-1}|]+ \prb\cdot(\E[|A_{i-1}|] + \E[|B_{i-1}|] + 1)$. Again, applying this equation successively, we obtain 
\[
\mathcal{E}=\E[|A_{r}|\cdot|B_{r}|]\le 2^{2 \lambda d}\prod_{i=1}^{r}\alpha_i + 4\cdot2^{\lambda d}\cdot\sum_{i=1}^{r}\left(\prod_{j=0}^{i-2}\alpha_{r-j}\right)\prb^{r-i+1} \le 2^{2 \lambda d}\prod_{i=1}^{r}\alpha_i + 4r\cdot2^{\lambda d}\cdot \prb^{r}\enspace.
\]\qed
\end{proof}

While \cref{lem:expectation-exact} gives an upper bound on the required expectation, it is not very handy. In the next lemma, we show how to further bound this expectation and how it affects the running time of the algorithm.

\begin{lemma}[Complexity for Arbitrary Distributions]
	Let $\calD$ be some distribution over $\F_2^d$, $r:=\frac{\lambda d}{\log^2 d}$, $\gamma\in\left[0,\frac{1}{2}\right] $ and $\lambda\in[0, 1]$.   Also let $\calE=\E[|A|\cdot|B|]$ for $A$ and $B$ in \linesref{line:naive-search} of \cref{alg:new-nn-it} when solving some instance of the $\NNP$ over $\calD$ (where the expectation is taken over the distribution of input lists and the random choices of the algorithm). Then \cref{alg:new-nn-it} solves the $\NNP$ over $\calD$ in time
	\[	
	\max\bigg({\prw^{-r}} ,\frac{2^{\lambda d}\cdot \prb^{r-1}}{\prw^r},\frac{2^{\varepsilon d}}{\prw^r}\bigg)^{1+o(1)}\enspace,
	\]
	where  
	\[
	\varepsilon=2 \lambda - \min\limits_{\substack{i\in[r]\\ \eta\in[0,1]}}(1-\eta)\left(1-\h\left(\frac{\delta - \frac{\eta}{2}}{1-\eta}\right)\right) - \frac{r\cdot\log p_{i, \eta k}}{d}
	\]
	with $p_{i, \eta k}:=\p{\wt(\block{\vec v + \vec w}{i}{r})=\eta k}$.
	\label{lem:expectation-symmetric}
\end{lemma}
\begin{proof}
Taking the result for $\mathcal{E}$ from \cref{lem:expectation-exact} and plugging into the run time formula from \cref{cor:arbitrary-dist} we get that the $\NNP$ problem over $\calD$ can be solved with probability overwhelming in $d$ in time 
\[\max\bigg({\prw^{-r}} ,\frac{2^{\lambda d}\cdot \prb^{r-1}}{\prw^r},\frac{\calE}{\prw^r}\bigg)^{1+o(1)}\le\max\bigg({\prw^{-r}} ,\frac{2^{\lambda d}\cdot \prb^{r-1}}{\prw^r},\frac{2^{2 \lambda d}\prod_{i=1}^{r}\alpha_i}{\prw^r}\bigg)^{1+o(1)}\]
since the right summand $\frac{4r\cdot2^{\lambda d}\cdot \prb^{r}}{\prw^r}$ of $\frac{\calE}{\prw^r}$ is asymptotically smaller than the second entry in the $\max$, i.e. $\frac{2^{\lambda d}\cdot \prb^{r-1}}{\prw^r}$. Thus, is suffices to find an easier upper bound for the first summand $S:=2^{2 \lambda d}\prod_{i=1}^{r}\alpha_i$. Remembering $\alpha_i=\p{\wt(\block{\vec v + \vec z}{i}{r}) = \delta k, \wt(\block{\vec w + \vec z}{i}{r})=\delta k}$ we receive

\begin{align*}
	S&\le 2^{2 \lambda d}\cdot \left(\max_{i\in[r]}\alpha_i\right)^r \\
	&=	2^{2 \lambda d}\cdot \left(\max_{i\in[r]}\sum_{j=0}^{k}\prw_{i, j} \cdot \p{\wt(\block{\vec v + \vec w}{i}{r})=j}\right)^r\\
	&\le	2^{2 \lambda d +o(d)}\cdot\left(\max_{i\in[r],\; j\in[k]\cup\{0\}}\prw_{i, j} \cdot \p{\wt(\block{\vec v + \vec w}{i}{r})=j}\right)^r\\
	&=	2^{2 \lambda d +o(d)}\cdot\left(\max_{i\in[r],\;\eta\in[0,1]}\prw_{i, \eta k} \cdot \p{\wt(\block{\vec v + \vec w}{i}{r})=\eta k}\right)^r\enspace,
\end{align*}
where $\prw_{i, \eta k}= \p{\wt(\block{\vec v + \vec z}{i}{r}) = \delta k, \wt(\block{\vec w + \vec z}{i}{r})=\delta k\mid \wt(\block{\vec v + \vec w}{i}{r})=\eta k}$. \cref{lem:xyinT} lets us rewrite this probability as
\[\prw_{i,\eta k}=\binom{\eta k}{\frac{1}{2}\eta k}\binom{(1-\eta) k}{\left(\delta - \frac{\eta}{2}\right)k}\left(\frac{1}{2}\right)^k\le2^{-\left(1-\h\left(\frac{\delta - \frac{\eta}{2}}{1-\eta}\right)\right)(1-\eta)k}\enspace.\]

We end up with

\begin{align*}
	S&\le 2^{2 \lambda d + r\cdot\max\limits_{i\in[r],\;\eta\in[0,1]}-\left(1-\h\left(\frac{\delta - \frac{\eta}{2}}{1-\eta}\right)\right)(1-\eta)k + \log p_{i, \eta k}+ o(d)} \\
	&=	2^{\left(2 \lambda + \max\limits_{i\in[r],\;\eta\in[0,1]}-\left(1-\h\left(\frac{\delta - \frac{\eta}{2}}{1-\eta}\right)\right)(1-\eta) + \frac{r}{d}\cdot\log p_{i, \eta k}\right)d+ o(d)}\\
	&=	2^{\left(2 \lambda - \min\limits_{i\in[r],\;\eta\in[0,1]}\left(1-\h\left(\frac{\delta - \frac{\eta}{2}}{1-\eta}\right)\right)(1-\eta) - \frac{r}{d}\cdot\log p_{i, \eta k}\right)d+ o(d)}\\
\end{align*}

with $p_{i, \eta k}:=\p{\wt(\block{\vec v + \vec w}{i}{r})=\eta k}$, which proves the claim.\qed
\end{proof}

Note that if it further holds that for $\vec v\sim \calD$ each of the $r$ blocks of $\vec v$ is identically distributed we can further simplify the term of $\varepsilon$ from \cref{lem:expectation-symmetric}. In this case, we have $p_{i, \eta k}^r \le \p{\wt(\vec v + \vec w)=\eta d}:=p_{\eta d}$, thus we get 
\[
\varepsilon=2\lambda-\min\limits_{\eta\in[0,1]} (1-\eta)\left(1-\h\left(\frac{\delta - \frac{\eta}{2}}{1-\eta}\right)\right)-\frac{\log p_{\eta d}}{d}\enspace.
\]

Now if we are given an arbitrary distribution $\calD$ we can maximize $\varepsilon$ according to $\eta$. Then we can similar to the proof of \cref{thm:main} derive a value for $\delta$ minimizing the overall time complexity.
\medskip 

We performed this maximization and optimization numerically for some generic input distributions. We considered distributions, where the weight of input vectors is distributed \emph{binomially}, chosen according to a \emph{Poisson} distribution or \emph{fixed} to a specific value. This means, first a weight is sampled according to the chosen distribution and then a vector of that weight is selected uniformly among all vectors of that weight.

The running time of \cref{alg:new-nn-it} for solving the $\NNP$ over the considered distributions seems to be only dependent on the expected weight of vectors contained in the input lists. That means the time complexity for input lists containing random vectors whose weight is either fixed to $\eta d$  or binomially or Poisson distributed with expectation $\eta d$ is equal. This can possibly be explained by the low variance of all these distributions, which implies a high concentration around this expected weight.

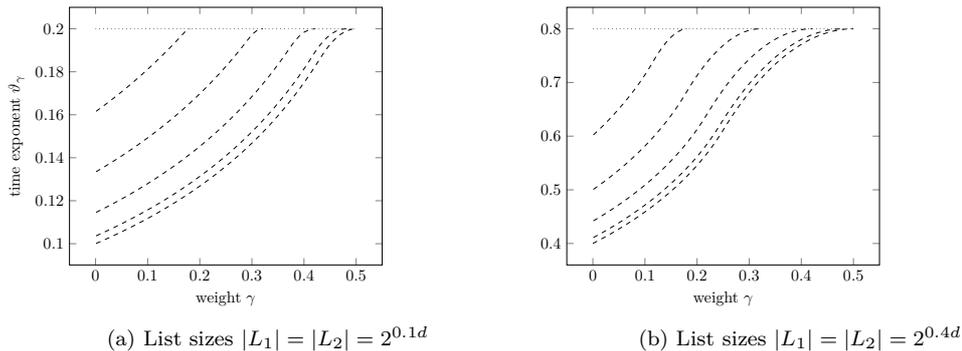
\begin{figure}
	\centering
	\begin{subfigure}{.5\textwidth}
		
		\begin{tikzpicture}[scale=0.6]
			\begin{axis}[
				y tick label style={
					/pgf/number format/.cd,
					fixed,
					precision=3,
					/tikz/.cd
				},
				x tick label style={
					/pgf/number format/.cd,
					fixed,
					1000 sep={},
					precision=2,
					/tikz/.cd
				},
				xlabel={weight $\gamma$},
				ylabel={time exponent $\vartheta_\gamma$},
				legend cell align={left},
				]
				
				\addplot [domain=0:0.5, samples=500,mark=none,black,	mark size=1.8, dotted, opacity=0.8] (x,0.2)  node[pos=0.5,yshift=4pt,opacity=0.8,xshift=-55pt] {};
				
				\pgfplotstableread{plots/plot_e_01}
				\cTradeOff
				\addplot[color=black,dashed] table[x = W,y=T] from \cTradeOff ;
				
				\pgfplotstableread{plots/plot_e_02}
				\cTradeOff
				\addplot[color=black,dashed] table[x = W,y=T] from \cTradeOff ;
				
				\pgfplotstableread{plots/plot_e_03}
				\cTradeOff
				\addplot[color=black,dashed] table[x = W,y=T] from \cTradeOff ;

				\pgfplotstableread{plots/plot_e_04}
				\cTradeOff
				\addplot[color=black,dashed] table[x = W,y=T] from \cTradeOff ;
				
				\pgfplotstableread{plots/plot_e_05}
				\cTradeOff
				\addplot[color=black,dashed] table[x = W,y=T] from \cTradeOff ;
				
			\end{axis}	
			
		\end{tikzpicture}
		\caption{List sizes $|L_1|=|L_2|=2^{0.1 d}$}
		\label{fig:diff-dist1}
	\end{subfigure}%
	\begin{subfigure}{.5\textwidth}
		
		\begin{tikzpicture}[scale=0.6]
			\begin{axis}[
				y tick label style={
					/pgf/number format/.cd,
					fixed,
					precision=3,
					/tikz/.cd
				},
				x tick label style={
					/pgf/number format/.cd,
					fixed,
					1000 sep={},
					precision=2,
					/tikz/.cd
				},
				xlabel={weight $\gamma$},
				ylabel={},
				legend cell align={left},
				]
				
				\addplot [domain=0:0.5, samples=500,mark=none,black,	mark size=1.8, dotted, opacity=0.8] (x,0.8)  node[pos=0.5,yshift=4pt,opacity=0.8,xshift=-55pt] {};
				
				\pgfplotstableread{plots/plot_e_0108}
				\cTradeOff
				\addplot[color=black,dashed] table[x = W,y=T] from \cTradeOff ;
				
				\pgfplotstableread{plots/plot_e_0208}
				\cTradeOff
				\addplot[color=black,dashed] table[x = W,y=T] from \cTradeOff ;
				
				\pgfplotstableread{plots/plot_e_0308}
				\cTradeOff
				\addplot[color=black,dashed] table[x = W,y=T] from \cTradeOff ;

				\pgfplotstableread{plots/plot_e_0408}
				\cTradeOff
				\addplot[color=black,dashed] table[x = W,y=T] from \cTradeOff ;
				
				\pgfplotstableread{plots/plot_e_0508}
				\cTradeOff
				\addplot[color=black,dashed] table[x = W,y=T] from \cTradeOff ;
				
			\end{axis}	
			
		\end{tikzpicture}
		
		\caption{List sizes $|L_1|=|L_2|=2^{0.4 d}$}
		\label{fig:diff-dist2}
	\end{subfigure}
	\caption{Time complexity exponents as a function of the weight of the closest pair for different input list distributions, where the expected weight of input elements is equal to $0.1d$, $0.2d$, $0.3d$, $0.4d$, $0.5d$ from left to right.}
	\label{fig:diff-dist}
\end{figure}

We see in \cref{fig:diff-dist}, that the value for $\gamma$, from where on the complexity becomes quadratic in the lists sizes shifts to the left. This behavior stems from the fact, that the expected weight of a sum of elements is no longer $\frac{d}{2}$, but roughly $2\eta(1-\eta)d$. What also stands out is, that the complexity for $\gamma=0$ is no longer linear in the lists sizes. The reason for this is that the probability of random pairs falling into the same bucket and the probability of the closest pair falling into the same bucket converge for decreasing weight of input list elements. This indicates that for input distributions with smaller expected weight a different bucketing criterion might be beneficial. We pose this as an open question for further research.

\section{Practical Experiments}
In this section, we give experimental results of the performance of a  proof of concept implementation of our new algorithm. These experiments verify the performance gain of our algorithm over a naive quadratic search approach. We also verify the numerical estimates of the algorithm's performance on different input distributions from the previous section and give some practical related improvements to our algorithm. Our implementation is publicly available at \url{https://github.com/FloydZ/NNAlgorithm}.
\smallskip

Before discussing the benchmark results let us first briefly describe some of the practical improvements we introduced in our implementation, which differ from the description in \cref{sec:new-alg}. 
We implemented a true depth-first search rather than the iterative description given previously. The iterative description just allowed for a more convenient analysis. Thus, our algorithm needs to store only the lists of a single path from the root to a leaf node at any time. Also, as all lists of subsequent levels are subsets of previous ones, we do not create $r$ different lists. We rather rearrange the elements of the input list such that elements belonging to the list of the subsequent level are consecutive, making it sufficient to just memorize the range of elements that belong to the next level list. This way, we only need to store the input list plus two integer markers for each level.
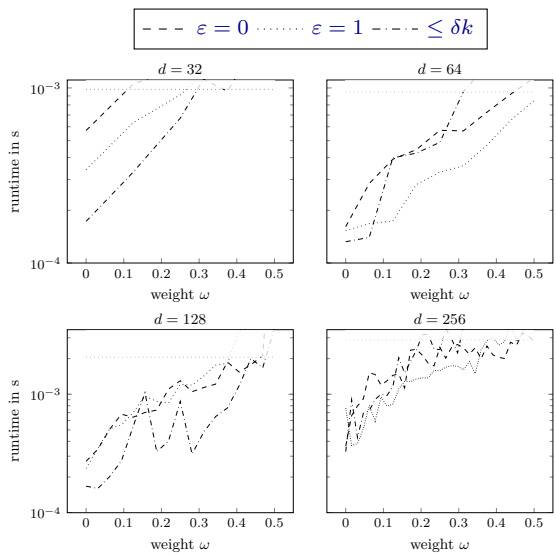
\begin{figure}[ht]
	\centering
	\begin{tikzpicture}[scale=0.63]
		\pgfplotsset{footnotesize}
		\begin{groupplot}[group style = {rows=2, columns=2, horizontal sep=20pt, vertical sep=40pt }, 
				width=0.45\textwidth,
				ymode=log,
				xlabel={weight $\omega$},
				xtick pos=left,
				ymin=0.0001,ymax=0.00111
			]
		\nextgroupplot[ylabel={runtime in s}, title = {$d=32$},
			legend style = { column sep = 1pt, legend columns = -1, legend to name = grouplegend2,}]
			
			\path[name path=axis] (axis cs:0,0.0011) -- (axis cs:0.5,0.0011);
			\addplot+[name path=A, domain=0:0.5, samples=500,mark=none,black, mark size=1.8, dotted, opacity=0.8, forget plot] (x,0.00098446) node[pos=0.5,yshift=4pt,opacity=0.8,xshift=-5pt] {};	
			
			
			\pgfplotstableread{plots/benchmarks/plot_n32_lam10_e0}
			\cTradeOff
			\addplot[color=black,dashed] table[x=W,y=T] from \cTradeOff ;
			\addlegendentry{$\varepsilon = 0$}
			
			\pgfplotstableread{plots/benchmarks/plot_n32_lam10_e1}				
			\cTradeOff
			\addplot[color=black,dotted] table[x=W,y=T] from \cTradeOff ;
			\addlegendentry{$\varepsilon = 1$}
			
			\pgfplotstableread{plots/benchmarks/plot_n32_lam10_allw_alldelta}				
			\cTradeOff
			\addplot[color=black,dashdotted] table[x=W,y=T] from \cTradeOff ;
			\addlegendentry[dashdotted]{$ \leq \delta k$}
			
			\tikzfillbetween[of=A and axis]{white, opacity=0.8};
			
		\nextgroupplot[title = {$d=64$}, yticklabels={,,}, ]
			\path[name path=axis] (axis cs:0,0.0011) -- (axis cs:0.5,0.0011);
			\addplot+[name path=A, domain=0:0.5, samples=500,mark=none,black, mark size=1.8, dotted, opacity=0.8, forget plot] (x,0.00095155)  node[pos=0.5,yshift=4pt,opacity=0.8,xshift=-5pt] {};

			\pgfplotstableread{plots/benchmarks/plot_n64_lam10_e0}
			\cTradeOff
			\addplot[color=black,dashed] table[x=W,y=T] from \cTradeOff ;
			
			\pgfplotstableread{plots/benchmarks/plot_n64_lam10_e1}				
			\cTradeOff
			\addplot[color=black,dotted] table[x=W,y=T] from \cTradeOff ;
			
			\pgfplotstableread{plots/benchmarks/plot_n64_lam10_allw_alldelta}				
			\cTradeOff
			\addplot[color=black,dashdotted] table[x=W,y=T] from \cTradeOff ;
			
			\tikzfillbetween[of=A and axis]{white, opacity=0.8};
			
		\nextgroupplot[ylabel={runtime in s}, title = {$d=128$}, ymin=0.0001,ymax=0.00351, ] 
			
			\path[name path=axis] (axis cs:0,0.0035) -- (axis cs:0.5,0.0035);
			\addplot+[name path=A, domain=0:0.5, samples=500,mark=none,black, mark size=1.8, dotted, opacity=0.8, forget plot] (x,0.002071)  node[pos=0.5,yshift=4pt,opacity=0.8,xshift=-55pt] {};	
			
			\pgfplotstableread{plots/benchmarks/plot_n128_lam10_e0}
			\cTradeOff
			\addplot[color=black,dashed] table[x=W,y=T] from \cTradeOff ;
			
			\pgfplotstableread{plots/benchmarks/plot_n128_lam10_e1}				
			\cTradeOff
			\addplot[color=black,dotted] table[x=W,y=T] from \cTradeOff ;
			
			\pgfplotstableread{plots/benchmarks/plot_n128_lam10_allw_alldelta}				
			\cTradeOff
			\addplot[color=black,dashdotted] table[x=W,y=T] from \cTradeOff ;
			
			\tikzfillbetween[of=A and axis]{white, opacity=0.8};
			
		\nextgroupplot[title = {$d=256$}, yticklabels={,,}, ymin=0.0001,ymax=0.00351, ] 
			
			\path[name path=axis] (axis cs:0,0.0035) -- (axis cs:0.5,0.0035);
			\addplot+[name path=A, domain=0:0.5, samples=500,mark=none,black, mark size=1.8, dotted, opacity=0.8, forget plot] (x,0.00289542)  node[pos=0.5,yshift=4pt,opacity=0.8,xshift=-55pt] {};	
			
			\pgfplotstableread{plots/benchmarks/plot_n256_lam10_e0}
			\cTradeOff
			\addplot[color=black,dashed] table[x=W,y=T] from \cTradeOff ;
			
			\pgfplotstableread{plots/benchmarks/plot_n256_lam10_e1}				
			\cTradeOff
			\addplot[color=black,densely dotted] table[x=W,y=T] from \cTradeOff ;
			
			\pgfplotstableread{plots/benchmarks/plot_n256_lam10_allw_alldelta}				
			\cTradeOff
			\addplot[color=black,dashdotted] table[x=W,y=T] from \cTradeOff ;
			
			\tikzfillbetween[of=A and axis, soft clip={domain=0:1},]{white, opacity=0.8};
		\end{groupplot}
		\node at ($(group c2r1) + (-2.70,3cm)$) {\ref{grouplegend2}}; 
	\end{tikzpicture}
	\caption{Runtime in seconds in logarithmic scale (y-axis) as a function of distance $\omega$ of the closest pair (x-axis) on random lists of size $2^{10}$. Dotted, dashed and dash-dotted lines indicate results for different bucketing, straight horizontal line is the time used by a naive quadratic search.}
\label{fig:diff-benchmark-10}
\end{figure}Also, it turns out that in practice often a small depth of the tree (not exceeding 8 in our experiments) is already sufficient to achieve good runtime results. 
Regarding the branching factor $N$ of the tree, we achieve optimal results either for values close to its expectation $\frac{1}{\prw}$ as given by the analysis or values being significantly smaller. The case of using a very small branching factor can be seen as a pruning strategy, similar to the one used in lattice enumeration algorithms for shortest vector search \cite{aono2018lower}. Additionally, we benchmarked three different strategies for the weight criteria: 
\begin{enumerate}
	\item Strictly enforcing a weight of $\delta k$ in each block, as described in our algorithm.
	\item Allowing for a small deviation $\pm \varepsilon$ around $\delta k$. 
	\item Allowing for weights of \emph{at most} $\delta k$.
\end{enumerate}
Further, we introduced a threshold for the size of the lists in the tree, below which the computation of further leaves is aborted and naive search is used instead.

\cref{fig:diff-benchmark-10} shows the runtime results for the different bucket criteria on small input lists of size $2^{10}$  containing random elements. Here, each data point was averaged over 50 measurements. The experimental results clearly indicate a significant gain over the quadratic search approach. The less significant gain for small dimension $d$ is due to the reduced amount of possible blocks or equivalently the low depth of the computation tree, which lets the algorithm not reach its full potential. In the case of small input lists, we observe that a bucketing strategy that allows a deviation of $\varepsilon=1$ from $\delta k$ is beneficial for most values of $d$.

\begin{figure}[ht]
	\centering
	\begin{tikzpicture}[scale=0.63]
		\pgfplotsset{footnotesize}
		\begin{groupplot}[group style = {rows=2, columns=2, horizontal sep=20pt, vertical sep=40pt }, 
								width=0.45\textwidth,
								ymode=log,
								xlabel={weight $\omega$},
								xtick pos=left,
								ymin=0.008,ymax=5.001
								]
		\nextgroupplot[ylabel={runtime in s}, title = {$d=32$},
			legend style = { column sep = 1pt, legend columns = -1, legend to name = grouplegend1,}]
			
			\path[name path=axis] (axis cs:0,1.5) -- (axis cs:0.5,1.5);
			\addplot+[name path=A, domain=0:0.5, samples=500,mark=none,black, mark size=1.8, dotted, opacity=0.8, forget plot] (x,0.819578) node[pos=0.5,yshift=4pt,opacity=0.8,xshift=-5pt] {};	
			
			
			\pgfplotstableread{plots/benchmarks/plot_n32_lam15_e0}
			\cTradeOff
			\addplot[color=black,dashed] table[x=W,y=T] from \cTradeOff ;
			\addlegendentry{$\varepsilon = 0$}
			
			\pgfplotstableread{plots/benchmarks/plot_n32_lam15_e1}				
			\cTradeOff
			\addplot[color=black,dotted] table[x=W,y=T] from \cTradeOff ;
			\addlegendentry{$\varepsilon = 1$}
			
			\pgfplotstableread{plots/benchmarks/plot_n32_lam15_allw_alldelta}				
			\cTradeOff
			\addplot[color=black,dashdotted] table[x=W,y=T] from \cTradeOff ;
			\addlegendentry[dashdotted]{$ \leq \delta k$}
			
			\tikzfillbetween[of=A and axis]{white, opacity=0.8};
		
		\nextgroupplot[title = {$d=64$}, yticklabels={,,},]
			\path[name path=axis] (axis cs:0,1.5) -- (axis cs:0.5,1.5);
			\addplot+[name path=A, domain=0:0.5, samples=500,mark=none,black, mark size=1.8, dotted, opacity=0.8, forget plot] (x,1.10301)  node[pos=0.5,yshift=4pt,opacity=0.8,xshift=-5pt] {};

			\pgfplotstableread{plots/benchmarks/plot_n64_lam15_e0}
			\cTradeOff
			\addplot[color=black,dashed] table[x=W,y=T] from \cTradeOff ;
			
			\pgfplotstableread{plots/benchmarks/plot_n64_lam15_e1}				
			\cTradeOff
			\addplot[color=black,dotted] table[x=W,y=T] from \cTradeOff ;
			
			\pgfplotstableread{plots/benchmarks/plot_n64_lam15_allw_alldelta}				
			\cTradeOff
			\addplot[color=black,dashdotted] table[x=W,y=T] from \cTradeOff ;
			
			\tikzfillbetween[of=A and axis]{white, opacity=0.8};
		
		\nextgroupplot[ylabel={runtime in s}, title = {$d=128$}, ymin=0.008,ymax=6.0001,] 
			
			\path[name path=axis] (axis cs:0,4) -- (axis cs:0.5,4);
			\addplot+[name path=A, domain=0:0.5, samples=500,mark=none,black, mark size=1.8, dotted, opacity=0.8, forget plot] (x,2.02504)  node[pos=0.5,yshift=4pt,opacity=0.8,xshift=-55pt] {};	
			
			\pgfplotstableread{plots/benchmarks/plot_n128_lam15_e0}
			\cTradeOff
			\addplot[color=black,dashed] table[x=W,y=T] from \cTradeOff ;
			
			\pgfplotstableread{plots/benchmarks/plot_n128_lam15_e1}				
			\cTradeOff
			\addplot[color=black,dotted] table[x=W,y=T] from \cTradeOff ;
			
			\pgfplotstableread{plots/benchmarks/plot_n128_lam15_allw_alldelta}				
			\cTradeOff
			\addplot[color=black,dashdotted] table[x=W,y=T] from \cTradeOff ;
			
			\tikzfillbetween[of=A and axis]{white, opacity=0.8};
			
		\nextgroupplot[title = {$d=256$}, yticklabels={,,}, ymin=0.008,ymax=4.0001, ] 
			
			\path[name path=axis] (axis cs:0,4) -- (axis cs:0.5,4);
			\addplot+[name path=A, domain=0:0.5, samples=500,mark=none,black, mark size=1.8, dotted, opacity=0.8, forget plot] (x,3.06437)  node[pos=0.5,yshift=4pt,opacity=0.8,xshift=-55pt] {};	
			
			\pgfplotstableread{plots/benchmarks/plot_n256_lam15_e0}
			\cTradeOff
			\addplot[color=black,dashed] table[x=W,y=T] from \cTradeOff ;
			
			\pgfplotstableread{plots/benchmarks/plot_n256_lam15_e1}				
			\cTradeOff
			\addplot[color=black,densely dotted] table[x=W,y=T] from \cTradeOff ;
			
			\pgfplotstableread{plots/benchmarks/plot_n256_lam15_allw_alldelta_v2}				
			\cTradeOff
			\addplot[color=black,dashdotted] table[x=W,y=T] from \cTradeOff ;
			
			\tikzfillbetween[of=A and axis, soft clip={domain=0:1},]{white, opacity=0.8};
		\end{groupplot}
		\node at ($(group c2r1) + (-2.70,3cm)$) {\ref{grouplegend1}}; 
	\end{tikzpicture}
		\caption{Runtime in seconds in logarithmic scale (y-axis) as a function of distance $\omega$ of the closest pair (x-axis) on random lists of size $2^{15}$. Dotted, dashed and dash-dotted lines indicate results for different bucketing strategies, straight horizontal line is time used by a naive quadratic search.}
\label{fig:diff-benchmar-15}
\end{figure}
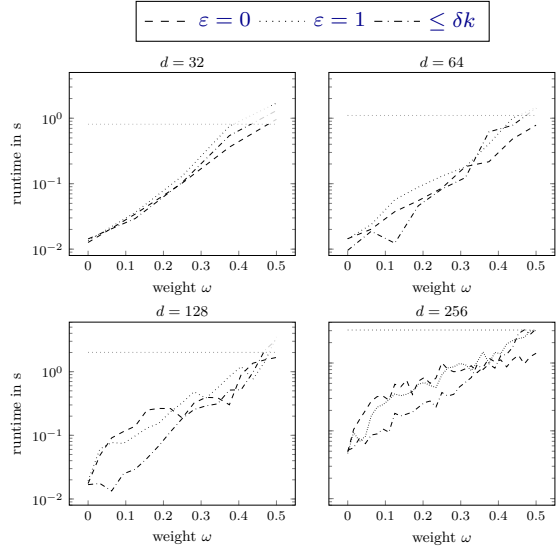 

\cref{fig:diff-benchmar-15} shows the same experiments performed on larger input lists of size $2^{15}$. Besides a more significant improvement over the naive search, we can observe that the bucketing criterion that uses $\delta k$ as an upper bound becomes more beneficial for nearly all values of $\gamma$ and $d$.

Eventually, \cref{fig:diff-benchmark-opt-gamma} shows the experimental runtime results on input lists, whose elements are drawn from a different input distribution, analyzed in \cref{sec:distributions}. Here the distribution is the uniformly random distribution over vectors of weight $\eta d$. One can observe that for growing $d$ the shape of the graph resembles the theoretical results from \cref{fig:diff-dist}. 
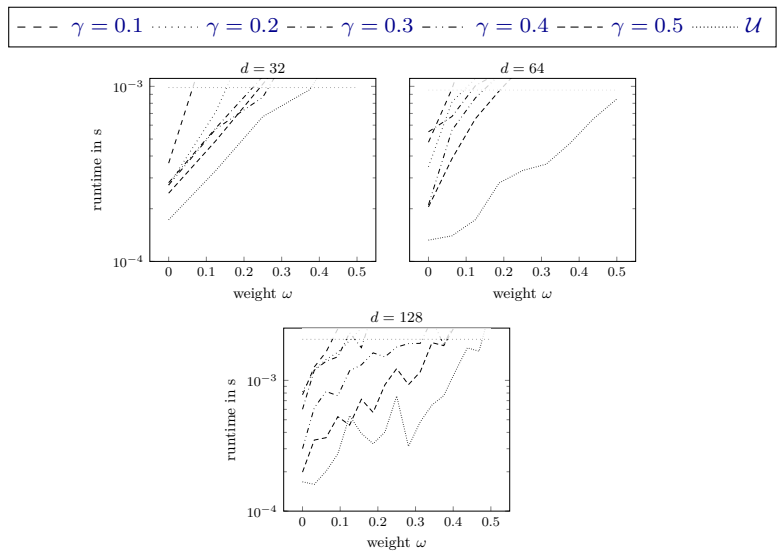
\begin{figure}[H]
	\centering
\begin{tikzpicture}[scale=0.63]
	\pgfplotsset{footnotesize}
	\begin{groupplot}[group style = {rows=2, columns=2, horizontal sep=20pt, vertical sep=40pt }, 
									width=0.45\textwidth,
									ymode=log,
									xlabel={weight $\omega$},
									xtick pos=left,
									ymin=0.0001,
									ymax=0.00111]
		\nextgroupplot[ylabel={runtime in s}, title = {$d=32$},
				legend style = { column sep = 1pt, legend columns = -1, legend to name = grouplegend3,}]

				\path[name path=axis] (axis cs:0,0.0011) -- (axis cs:0.5,0.0011);
				\addplot+[name path=A, domain=0:0.5, samples=500,mark=none,black, mark size=1.8, dotted, opacity=0.8, forget plot] (x,0.00098446)  node[pos=0.5,yshift=4pt,opacity=0.8,xshift=-55pt] {};	
				
				\pgfplotstableread{plots/benchmarks/g/32/plot_opt_n32_lam10_gamma0.1}
				\cTradeOff
				\addplot[color=black,dashed] table[x=W,y=T] from \cTradeOff ;
				\addlegendentry{$\gamma = 0.1$}
				
				\pgfplotstableread{plots/benchmarks/g/32/plot_opt_n32_lam10_gamma0.2}				
				\cTradeOff
				\addplot[color=black,dotted] table[x=W,y=T] from \cTradeOff ;
				\addlegendentry{$\gamma = 0.2$}
				
				\pgfplotstableread{plots/benchmarks/g/32/plot_opt_n32_lam10_gamma0.3}
				\cTradeOff
				\addplot[color=black,dashdotted] table[x=W,y=T] from \cTradeOff ;
				\addlegendentry{$\gamma = 0.3$}
				
				\pgfplotstableread{plots/benchmarks/g/32/plot_opt_n32_lam10_gamma0.4}				
				\cTradeOff
				\addplot[color=black,dashdotdotted] table[x=W,y=T] from \cTradeOff ;
				\addlegendentry{$\gamma = 0.4$}
				
				\pgfplotstableread{plots/benchmarks/g/32/plot_opt_n32_lam10_gamma0.5}
				\cTradeOff
				\addplot[color=black,densely dashed] table[x=W,y=T] from \cTradeOff ;
				\addlegendentry{$\gamma = 0.5$}
				
				\pgfplotstableread{plots/benchmarks/plot_opt_n32_lam10}
				\cTradeOff
				\addplot[color=black,densely dotted] table[x=W,y=T] from \cTradeOff ;
				\addlegendentry{$\mathcal{U}$}
				
				\tikzfillbetween[of=A and axis, soft clip={domain=0:1},]{white, opacity=0.8}

		\nextgroupplot[title = {$d=64$}, yticklabels={,,}, ]
				\path[name path=axis] (axis cs:0,0.0011) -- (axis cs:0.5,0.0011);
				\addplot+[name path=A, domain=0:0.5, samples=500,mark=none,black, mark size=1.8, dotted, opacity=0.8, forget plot] (x,0.00095155)  node[pos=0.5,yshift=4pt,opacity=0.8,xshift=-55pt] {};	
				
				\pgfplotstableread{plots/benchmarks/g/64/plot_opt_n64_lam10_gamma0.1_v2}
				\cTradeOff
				\addplot[color=black,dashed] table[x=W,y=T] from \cTradeOff ;
				
				\pgfplotstableread{plots/benchmarks/g/64/plot_opt_n64_lam10_gamma0.2_v2}				
				\cTradeOff
				\addplot[color=black,dotted] table[x=W,y=T] from \cTradeOff ;
				
				\pgfplotstableread{plots/benchmarks/g/64/plot_opt_n64_lam10_gamma0.3_v2}
				\cTradeOff
				\addplot[color=black,dashdotted] table[x=W,y=T] from \cTradeOff ;
				
				\pgfplotstableread{plots/benchmarks/g/64/plot_opt_n64_lam10_gamma0.4_v2}				
				\cTradeOff
				\addplot[color=black,dashdotdotted] table[x=W,y=T] from \cTradeOff ;
				
				\pgfplotstableread{plots/benchmarks/g/64/plot_opt_n64_lam10_gamma0.5_v2}
				\cTradeOff
				\addplot[color=black,densely dashed] table[x=W,y=T] from \cTradeOff ;
				
				\pgfplotstableread{plots/benchmarks/plot_opt_n64_lam10}
				\cTradeOff
				\addplot[color=black,densely dotted] table[x=W,y=T] from \cTradeOff ;

				\tikzfillbetween[of=A and axis]{white, opacity=0.8};
		\nextgroupplot[ylabel={runtime in s}, title = {$d=128$}, xshift=1.1in, ymin=0.0001, ymax=0.00251, ] 

				\path[name path=axis] (axis cs:0,0.0025) -- (axis cs:0.5,0.0025);
				\addplot+[name path=A, domain=0:0.5, samples=500,mark=none,black, mark size=1.8, dotted, opacity=0.8, forget plot] (x,0.002071)  node[pos=0.5,yshift=4pt,opacity=0.8,xshift=-55pt] {};	
				
				\pgfplotstableread{plots/benchmarks/g/128/plot_opt_n128_lam10_gamma0.1}
				\cTradeOff
				\addplot[color=black,dashed] table[x=W,y=T] from \cTradeOff ;
				
				\pgfplotstableread{plots/benchmarks/g/128/plot_opt_n128_lam10_gamma0.2}				
				\cTradeOff
				\addplot[color=black,dotted] table[x=W,y=T] from \cTradeOff ;
				
				\pgfplotstableread{plots/benchmarks/g/128/plot_opt_n128_lam10_gamma0.3}
				\cTradeOff
				\addplot[color=black,dashdotted] table[x=W,y=T] from \cTradeOff ;
				
				\pgfplotstableread{plots/benchmarks/g/128/plot_opt_n128_lam10_gamma0.4}				
				\cTradeOff
				\addplot[color=black,dashdotdotted] table[x=W,y=T] from \cTradeOff ;
				
				\pgfplotstableread{plots/benchmarks/g/128/plot_opt_n128_lam10_gamma0.5}
				\cTradeOff
				\addplot[color=black,densely dashed] table[x=W,y=T] from \cTradeOff ;
				
				\pgfplotstableread{plots/benchmarks/plot_opt_n128_lam10_v2}
				\cTradeOff
				\addplot[color=black,densely dotted] table[x=W,y=T] from \cTradeOff ;
			
				\tikzfillbetween[of=A and axis, soft clip={domain=0:1},]{white, opacity=0.8};
	\end{groupplot}
	\node at ($(group c2r1) + (-2.72,3cm)$) {\ref{grouplegend3}}; 
\end{tikzpicture}
	\caption{Runtime in seconds in logarithmic scale (y-axis) as a function of distance $\omega$ of the closest pair (x-axis) on lists of size $2^{10}$ containing random elements of weight $\gamma d$. The densely dashed line ($\calU$) indicates the runtime on uniformly random lists.}
\label{fig:diff-benchmark-opt-gamma}
\end{figure}

\bibliographystyle{splncs03}
\bibliography{cryptobib/abbrev0,cryptobib/crypto,cryptobib/manual}

%
\end{document}